\documentclass[english,runningheads,12pt]{llncs}
\usepackage{tabularx,booktabs,multirow,delarray,array}
\usepackage{graphicx,amssymb,amsmath,amssymb}
\usepackage{latexsym}

\usepackage[in]{fullpage}

\def\calA{\mathcal{A}}

\def\calF{\mathcal{F}}
\def\calL{\mathcal{L}}
\def\calR{\mathcal{R}}
\def\calT{\mathcal{T}}
\def\bbR{\mathbb{R}}

\newcommand{\eat}[1]{}

\begin{document}

\title{Range Queries on Uncertain Data\thanks{A preliminary version of
this paper appeared in the Proceedings of the 25th International Symposium
on Algorithms and Computation (ISAAC 2014).
}
}

\author{Jian Li \inst{1}
\and
Haitao Wang
\inst{2}
}
\institute{
Institute for Interdisciplinary Information Sciences\\
Tsinghua University, Beijing 100084, China.\\
\email{lijian83@mail.tsinghua.edu.cn}\\
\and
Department of Computer Science\\
Utah State University, Logan, UT 84322, USA\\
\email{haitao.wang@usu.edu}
}

\date{}

\maketitle

\thispagestyle{empty}

\pagestyle{plain}
\pagenumbering{arabic}

\vspace{-0.2in}
\begin{abstract}
Given a set $P$ of $n$ uncertain points on
the real line, each represented by its one-dimensional probability density function, we consider the problem of building data structures on $P$ to answer range
queries of the following three types for any query interval $I$:
(1) top-$1$ query: find the point in $P$ that lies in $I$ with the highest
probability, (2) top-$k$ query: given any integer $k\leq n$ as part of the query,
return the $k$ points in $P$ that lie in $I$ with the highest probabilities, and (3) threshold query: given any threshold $\tau$ as part of the query, return all points
of $P$ that lie in $I$ with probabilities at least $\tau$.
We present data structures for these range queries with linear or nearly linear space and efficient query time.
\end{abstract}


\section{Introduction}
\label{sec:intro}

With a rapid increase in the number of application domains,
such as data integration, information
extraction, sensor networks, scientific measurements etc., where uncertain data are generated
in an unprecedented speed, managing, analyzing and query processing over such data has
become a major challenge and have received significant attentions. We study one important problem in this domain, building data structures for uncertain
data for efficiently answering certain range queries. The problem has been
studied extensively with a wide range of applications \cite{ref:AgarwalIn09,ref:ChengEf04,ref:KnightEf11,ref:LjosaAP07,ref:QiTh10,ref:TaoIn05,ref:TaoRa07}.
We formally define the problems below.

Let $\bbR$ be any real line (e.g., the $x$-axis).
In the (traditional) deterministic version of this problem,
we are given a set $P$ of $n$ deterministic points on $\bbR$, and the
goal is to build a data structure (also called ``index'' in database) such that
given a range, specified by an interval $I\subseteq \bbR$, one point (or all points)
in $I$ can be retrieved efficiently.
It is well known that a simple solution for this problem is a binary search tree over all points
which is of linear size and can support logarithmic (plus output size) query time.
However, in many applications, the location of each point may be uncertain
and the uncertainty is represented in the form of probability distributions
\cite{ref:AgarwalNe12,ref:AgrawalTr06,ref:ChengEf04,ref:TaoIn05,ref:TaoRa07}.
In particular, an {\em uncertain point} $p$ is specified by its
probability density function (pdf) $f_p: \bbR\rightarrow\bbR^+\cup\{0\}$.
Let $P$ be the set of $n$ uncertain points in $\bbR$ (with pdfs specified as input).
Our goal is to build data structures to quickly
answer range queries on $P$. In this paper, we consider the following
three types of range queries, each of which involves a query interval
$I=[x_l,x_r]$. For any point $p\in P$, we use $\Pr[p\in I]$ to denote the probability that $p$ is contained in $I$.

\begin{description}
\item[Top-1 query:]
Return the point $p$ of $P$ such that
$\Pr[p\in I]$ is the largest.
\item[Top-$k$ query:]
Given any integer $k$, $1\leq k\leq n$, as part of the query,
return the $k$ points $p$ of $P$ such that
$\Pr[p\in I]$ are the largest.
\item[Threshold query:]
Given a threshold $\tau$, as part of the query, return all points $p$ of $P$ such that $\Pr[p\in
I]\geq \tau$.
\end{description}

We assume $f_p$ is a step function, i.e., a {\em histogram} consisting of at most $c$
pieces (or intervals) for some integer $c\geq 1$ (e.g., see Fig.~\ref{fig:histogram}).
More specifically, $f_p(x)=y_i$ for $x_{i-1}\leq x<x_i$, $i=1,\ldots,c$, with
$x_0=-\infty$, $x_c=\infty$, and $y_1=y_c=0$.
Throughout the paper, we assume $c$ is a constant.
The cumulative distribution function (cdf) $F_p(x)=\int_{-\infty}^xf_p(t)dt$ is a
monotone piecewise-linear function consisting of $c$ pieces (e.g., see
Fig.~\ref{fig:cdf}). Note that $F_p(+\infty)=1$, and for any interval $I=[x_l,x_r]$ the
probability $\Pr[p\in I]$ is $F_p(x_r)-F_p(x_l)$.
From a geometric point of view, each interval of $f_p$ defines a rectangle with the $x$-axis, and the sum of the areas of all these rectangles of $f_p$ is exactly one. Further, the cdf value $F_p(x)$ is the sum of the areas of the subsets of these rectangles to the left of the vertical line through $x$ (e.g., see Fig.~\ref{fig:cdfvalue}), and the probability $\Pr[p\in I]$ is the sum of the areas of the subsets of these rectangles between the two vertical lines through $x_l$ and $x_r$, respectively (e.g., see Fig.~\ref{fig:probvalue}).

As discussed in \cite{ref:AgarwalIn09}, the histogram model
can be used to approximate most pdfs with arbitrary
precision in practice. In addition, the {\em discrete} pdf where each
uncertain point can appear in a few
locations, each with a certain probability, can be viewed as a special
case of the histogram model because we can use infinitesimal pieces around these locations.

\begin{figure}[t]
\begin{minipage}[t]{0.45\linewidth}
\begin{center}
\includegraphics[totalheight=1.0in]{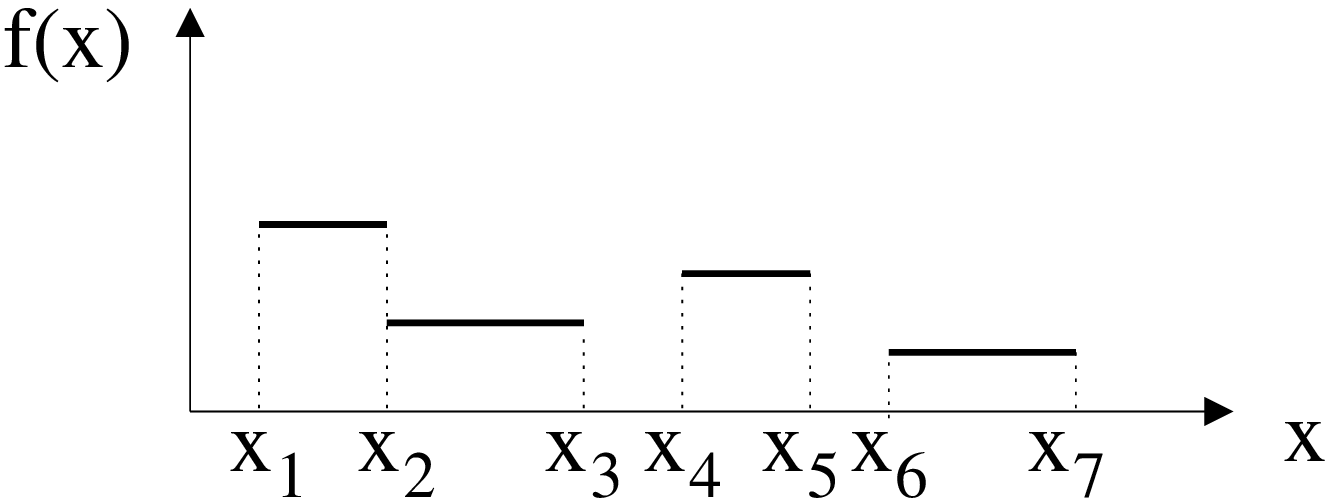}
\caption{\footnotesize The pdf of an uncertain point.}
\label{fig:histogram}
\end{center}
\end{minipage}
\hspace{0.02in}
\begin{minipage}[t]{0.54\linewidth}
\begin{center}
\includegraphics[totalheight=1.0in]{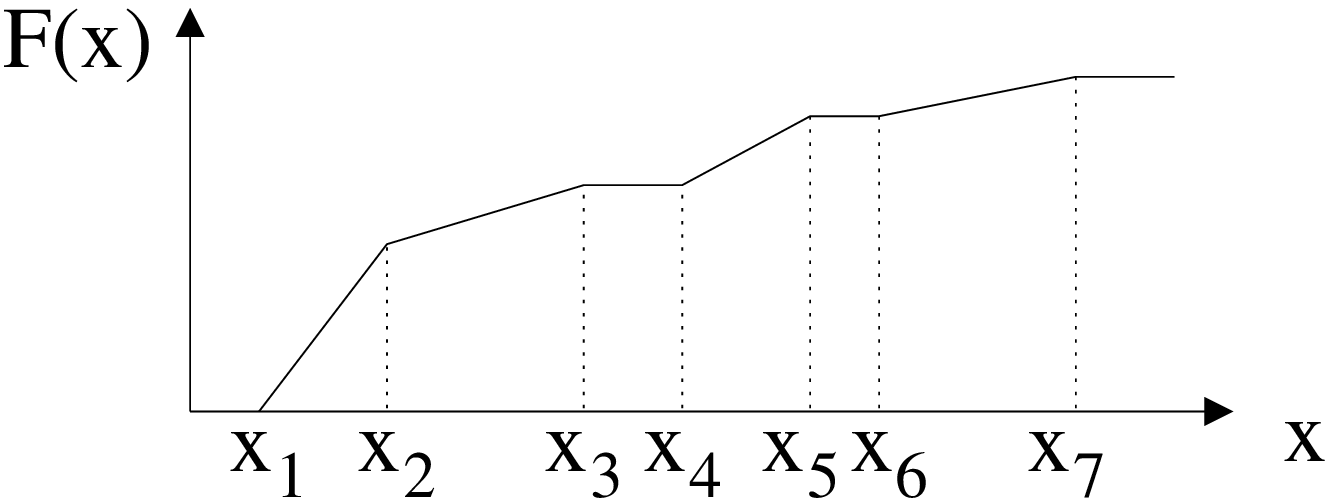}
\caption{\footnotesize The cdf of the uncertain point in
Fig.~\ref{fig:histogram}.}
\label{fig:cdf}
\end{center}
\end{minipage}
\vspace*{-0.20in}
\end{figure}

We also study an important special case where the pdf $f_p$ is a uniform distribution
function, i.e., $f$ is associated with an interval
$[x_l(p),x_r(p)]$ such that
$f_p(x)=1/(x_r(p)-x_l(p))$ if $x\in [x_l(p),x_r(p)]$ and $f_p(x)=0$
otherwise. Clearly, the cdf $F_p(x)=(x-x_{l}(p))/(x_{r}(p)-x_{l}(p))$ if
$x\in [x_{l}(p),x_{r}(p)]$, $F_p(x)=0$ if $x\in (-\infty,x_l(p))$, and
$F_p(x)=1$ if $x\in (x_r(p),+\infty)$.
Uniform distributions have been used as a major representation of uncertainty
in some previous work (e.g., \cite{icde08-probnn,ref:ChengEf04,li2010ranking}).
We refer to this special case the {\em uniform case} and
the more general case where $f_p$ is a histogram
distribution function as the {\em histogram} case.

Throughout the paper, we will always use $I=[x_l,x_r]$ to denote the query
interval.
The query interval $I$ is {\em unbounded} if
either $x_l=-\infty$ or $x_r=+\infty$ (otherwise, $I$ is {\em bounded}).
For the threshold query, we will always use $m$ to denote the output
size of the query, i.e., the number of
points $p$ of $P$ such that $\Pr[p\in I]\geq \tau$.

Range reporting on uncertain data has many applications \cite{ref:AgarwalIn09,ref:ChengEf04,ref:KnightEf11,ref:QiTh10,ref:TaoIn05,ref:TaoRa07},
As shown
in \cite{ref:AgarwalIn09}, our problems are also useful even in some
applications that involve only deterministic data. For example, consider the
movie rating system in IMDB where each reviewer gives a rating from 1
to 10. A top-$k$ query on $I=[7,+\infty)$
would find ``the $k$ movies such that the percentages of the
ratings they receive at least 7 are the largest''; a threshold query
on $I=[7,+\infty)$ and $\tau=0.85$ would
find ``all the movies such that at least 85\% of the ratings they
receive are larger than or equal to 7''.
Note that in the above examples the
interval $I$ is unbounded, and thus, it would also be interesting to have
data structures particularly for quickly answering
queries with unbounded query intervals.

\begin{figure}[t]
\begin{minipage}[t]{0.49\linewidth}
\begin{center}
\includegraphics[totalheight=1.2in]{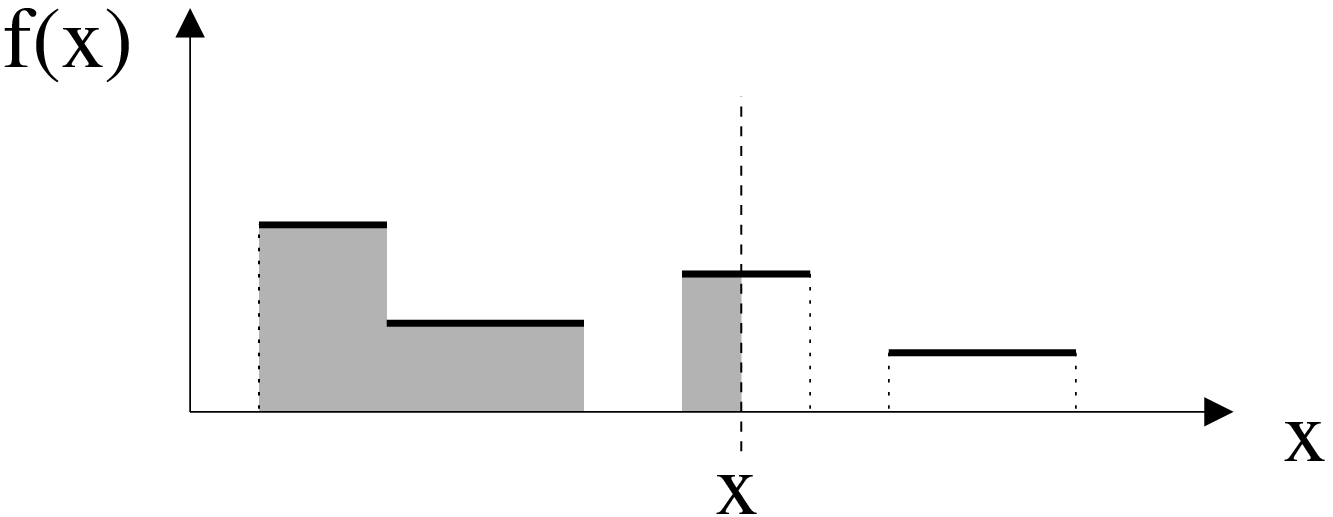}
\caption{\footnotesize Geometrically, $F_p(x)$ is equal to the sum of the areas of the shaded rectangles.}
\label{fig:cdfvalue}
\end{center}
\end{minipage}
\hspace{0.02in}
\begin{minipage}[t]{0.49\linewidth}
\begin{center}
\includegraphics[totalheight=1.2in]{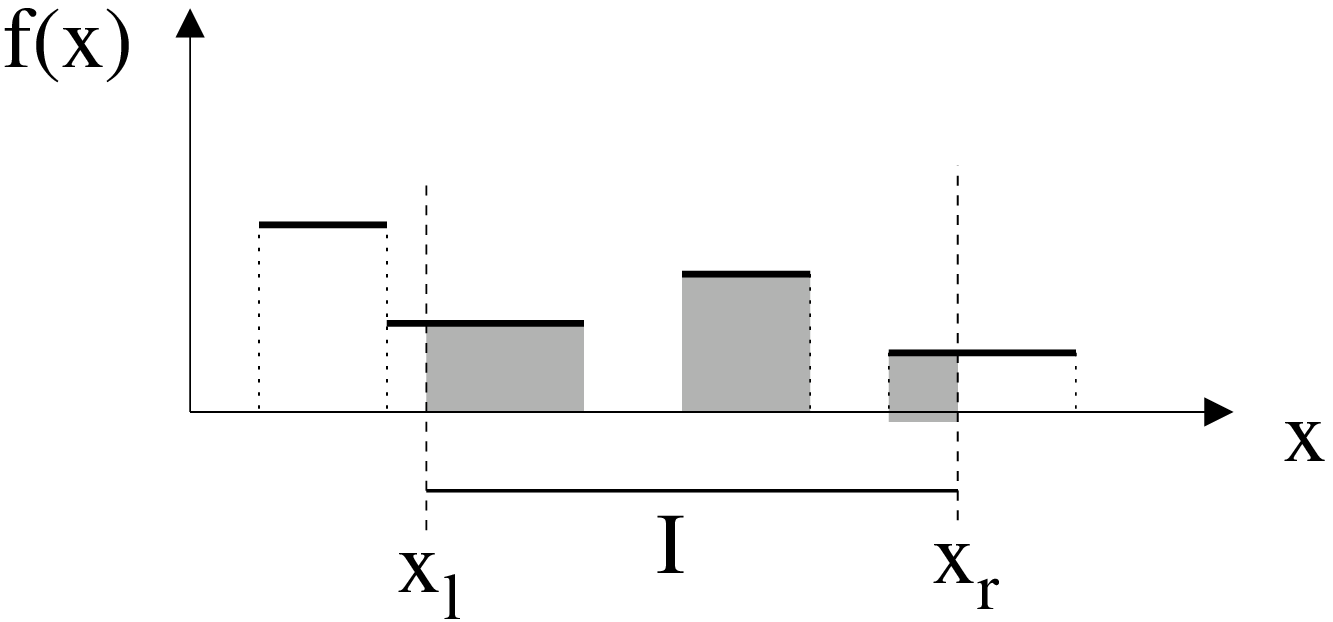}
\caption{\footnotesize Geometrically, the probability $\Pr[p\in I]$ is equal to the sum of the areas of the shaded rectangles.}
\label{fig:probvalue}
\end{center}
\end{minipage}
\vspace*{-0.20in}
\end{figure}

\subsection{Previous Work}
\label{sec:prework}

The threshold query was first introduced by Cheng {\em et al.}
\cite{ref:ChengEf04}.  Using R-trees, they \cite{ref:ChengEf04} gave
heuristic algorithms for the histogram case, without
 any theoretical performance guarantees. For the
uniform case, if $\tau$ is fixed for any query,
they proposed a data structure of $O(n\tau^{-1})$ size with
$O(\tau^{-1}\log n+m)$ query time  \cite{ref:ChengEf04}.
These bounds depend on $\tau^{-1}$, which can be arbitrarily large.

Agarwal {\em et al.} \cite{ref:AgarwalIn09} made a
significant theoretical step on solving the threshold queries for the
histogram case: If the threshold $\tau$ is fixed, their approach can build an $O(n)$ size
data structure in $O(n\log n)$ time, with $O(m+\log n)$
query time; if $\tau$ is not fixed, they built an
$O(n\log^2 n)$ size data structure in $O(n\log^3 n)$ expected time that can
answer each query in $O(m+\log^3 n)$ time.
Tao {\em et al.} \cite{ref:TaoIn05,ref:TaoRa07} considered the threshold queries
in two and higher dimensions. They provided heuristic results and
a query takes $O(n)$ time in the worst case.
Heuristic solutions were also given elsewhere, e.g.
\cite{ref:KnightEf11,ref:QiTh10,ref:SinghIn07}.
Recently, Abdullah {\em et al.} \cite{ref:AbdullahRa13} extended the notion of
{\em geometric coresets} to uncertain data
for range queries in order to obtain efficient approximate solutions.

Our work falls into the broad area of managing and analyzing uncertain data
which has attracted
significant attentions recently in database community.
This line of work has spanned
a range of issues from theoretical foundation of data models and data languages,
algorithmic problems for efficiently answering various queries, to
system implementation issues.
Probabilistic database systems have emerged as a major platform for this purpose and
several prototype systems have been built to address
different aspects/challenges in managing probabilistic data, e.g. MYSTIQ~\cite{dalvi:vldbj06}, Trio~\cite{ref:AgrawalTr06},
ORION~\cite{cheng:sigmod03}, MayBMS~\cite{koch2009maybms}, PrDB~\cite{sen:vldbj09}, MCDB~\cite{jampani2008mcdb}.
Besides of the range queries we mentioned above,
there has also been much work on efficiently answering different types of queries over probabilistic data, such as conjunctive queries (or the union of conjunctive queries)~\cite{dalvi:vldbj06,dalvi2012dichotomy},
aggregates~\cite{DBLP:conf/pods/JayramMMV07,conf/dbpl/re07}, top-$k$ and ranking~\cite{Cormode09,li2010ranking,li2011unified,re:icde07,soliman:icde07},
clustering~\cite{cormode2008approximation,guha2009exceeding},
nearest neighbors~\cite{ref:BeskalesEf08,icde08-probnn,potamias2010k}, and so on.
We refer interested readers to the recent book \cite{suciu2011probabilistic} for more information.

As discussed in \cite{ref:AgarwalIn09}, our uncertain model is an analogue of the
{\em attribute-level uncertainty model} in the probabilistic database literature.
Another popular model is the {\em tuple-level uncertainty
model} \cite{ref:AgrawalTr06,dalvi:vldbj06,ref:YiuEf09}, where a tuple has fixed attribute values but its existence
is uncertain. The range query under the latter model is much easier
since a $d$-dimensional range searching over uncertain data can be
transformed to a ($d+1$)-dimensional range searching problem over
certain data \cite{ref:AgarwalIn09,ref:YiuEf09}. In contrast, the
problem under the former model is more challenging, partly because it is
unclear how to transform it to an instance on certain data.

\subsection{Our Results}

Based on our above discussion, the problem has four variations: the {\em uniform unbounded} case where each pdf $f_p$ is a uniform distribution function and each query interval $I$ is unbounded, the {\em uniform bounded} case where each pdf $f_p$ is a uniform distribution function and each query interval $I$ is bounded, the {\em histogram unbounded} case where each pdf $f_p$ is a general histogram distribution function and each query interval $I$ is unbounded, and the {\em uniform bounded} case where each pdf $f_p$ is a general histogram distribution function and each query interval $I$ is bounded.
Refer to Table \ref{tab:results} for a summary of our results on the four cases.

Note that we also present solutions to the most general case (i.e., the histogram bounded case), which were originally left as open problems in the preliminary version of this paper \cite{ref:LiRa14}.

\begin{table}[t]
\begin{center}
{
\footnotesize
\begin{tabularx}{0.866\textwidth}{llllll}
\toprule
\multicolumn{3}{c}{Four Problem Variations} &Top-1 Queries\ \ \ \  & Top-$k$  Queries\ \ \ & Threshold  Queries \\
\addlinespace[1ex]
\midrule[0.01in]
\addlinespace[1ex]
\multirow{6}{*}{Unbounded\ } & \multirow{3}{*}{Uniform} &Preprocessing Time\ \ &$O(n\log n)$ & $O(n\log n)$& $O(n\log n)$\\
& & Space & $O(n)$ & $O(n)$  & $O(n)$ \\
& & Query Time & $O(\log n)$ \  & $O(\log n+k)$ \ \  & $O(\log n+m)$\\
\addlinespace[1ex]
\cline{2-6}
\addlinespace[1ex]
& \multirow{3}{*}{Histogram\ } & Preprocessing Time & $O(n\log n)$ & $O(n\log n)$ &$O(n\log n)$\\
& & Space & $O(n)$ &$O(n)$ &$O(n)$\\
& & Query Time & $O(\log n)$ & $T$ & $O(\log n+m)$ \\
\midrule[0.008in]
\multirow{6}{*}{Bounded} & \multirow{3}{*}{Uniform} & Preprocessing Time & $O(n\log n)$ & $O(n\log^2 n)$ & $O(n\log^2 n)$\\
& & Space & $O(n)$ & $O(n\log n)$ & $O(n\log n)$ \\
& & Query Time & $O(\log n)$ & $T$ & $O(\log n+m)$\\
\addlinespace[1ex]
\cline{2-6}
\addlinespace[1ex]
& \multirow{3}{*}{Histogram} & Preprocessing Time & $O(n\log^3 n)$ &  $O(n\log^3 n)^*$ & $O(n\log^3 n)^*$ \cite{ref:AgarwalIn09}\\
& & Space & $O(n\log^2 n)$ &  $O(n\log^2 n)$ &  $O(n\log^2 n)$ \cite{ref:AgarwalIn09}\\
& & Query Time & $O(\log^3 n)$ & $O(\log^3 n +k)$ \ \ \ \ \ \ &$O(\log^3 n+m)$ \cite{ref:AgarwalIn09}\\
\bottomrule
\end{tabularx}
\vspace*{0.1in}
\caption{\footnotesize Summary of our results (the result for threshold queries of the histogram bounded case is from \cite{ref:AgarwalIn09}): $T$ is $O(k)$ if $k=\Omega(\log n\log\log n)$ and $O(\log n+k\log k)$ otherwise. For threshold queries, $m$ is the output size of each query. All time complexities are deterministic except the preprocessing times for top-$k$ and threshold queries of the histogram bounded case (marked with *).
}
\label{tab:results}
}
\end{center}
\vspace*{-0.4in}
\end{table}

We say the {\em complexity} of a data structure is $O(A,B)$ if can be built in $O(A)$ time and its size is $O(B)$.

\begin{itemize}
\item
For the uniform unbounded case, the complexities of our data structures for the three types of queries are all $O(n\log n, n)$. The top-$1$ query time is
$O(\log n)$; the top-$k$ query time is $O(\log n+ k)$; the threshold query time is $O(\log n+m)$.
\item
For the histogram unbounded case, our results are the same as the above uniform unbounded case except that the time for each top-$k$ query is $O(k)$ if $k=\Omega(\log n\log\log n)$ and
$O(\log n+k\log k)$ otherwise (i.e., for large $k$, the algorithm has a better performance).
\item
For the uniform bounded case, the complexity of our top-1 data structure is $O(n\log n,
n)$, with query time $O(\log n)$. For the other two types of queries, the complexities of our data structures are both $O(n\log^2 n, n\log n)$; the top-$k$ query time is $O(k)$ if
$k=\Omega(\log n\log\log n)$ and $O(\log n+k\log k)$ otherwise,
and the threshold query time is $O(\log n+m)$.
\item
For the histogram bounded case, for threshold queries,
Agarwal {\em et al.} \cite{ref:AgarwalIn09} built a data structure of size
$O(n\log^2 n)$ in $O(n\log^3 n)$ expected time, with $O(\log^3 n + m)$ query time.
Note that our results on the threshold queries for the two uniform cases and the histogram unbounded case are clearly better than the above solution in \cite{ref:AgarwalIn09}.
For top-1 queries, we build a data structure of $O(n\log^2 n)$ size in $O(n\log^3 n)$ (deterministic) time, with $O(\log^3 n)$ query time. For top-$k$ queries, we build a data structure of $O(n\log^2 n)$ size in
$O(n\log^3 n)$ expected time, with $O(\log^3 n+k)$ query time.
\end{itemize}

Note that all above results are based on the assumption that $c$ is a constant;
otherwise these results still hold with replacing $n$ by $c\cdot n$
except that for the histogram bounded case the results hold with replacing $n$ by $c^2 n$.

%

The rest of the paper is organized as follows.
We first introduce the notations and some observations in Section \ref{sec:pre}. We present our results for the uniform case in
Section \ref{sec:uniform}. The histogram case is discussed in Section
\ref{sec:nonuniform}. We conclude the paper in Section \ref{sec:conclusions}.

\section{Preliminaries}
\label{sec:pre}

Recall that an {\em uncertain point} $p$ is specified by its
pdf $f_p: \bbR\rightarrow\bbR^+\cup\{0\}$ and
the corresponding cdf is $F_p(x)=\int_{-\infty}^xf_p(t)dt$ is a
monotone piecewise-linear function (with at most $c$ pieces).
For each uncertain point $p$, we call $\Pr[p\in I]$ the {\em
$I$-probability} of $p$.
Let $\calF$ be the set of the cdfs of all points of $P$. Since each
cdf is an increasing piecewise linear function,
depending on
the context, $\calF$ may also refer to the set of the $O(n)$ line segments
of all cdfs.
Recall that $I=[x_l,x_r]$ is the query
interval.
We start with an easy observation.

\begin{lemma}\label{lem:10}
If $x_l=-\infty$, then for any uncertain point $p$, $\Pr[p\in I]=F_p(x_r)$.
\end{lemma}
\begin{proof}
Due to $x_l=-\infty$, $\Pr[p\in I]=\int_{-\infty}^{x_r}f_p(t)dt$, which
is exactly $F_p(x_r)$.\qed
\end{proof}

Let $L$ be the vertical line with $x$-coordinate $x_r$. Since each cdf
$F_p$ is a monotonically increasing function, there is only one
intersection between $F_p$  and $L$. It is easy to
know that for each cdf $F_p$ of $\calF$, the $y$-coordinate of the
intersection of $F_p$ and $L$ is $F_p(x_r)$, which is the
$I$-probability of $p$ by Lemma \ref{lem:10}. For each point in any cdf of $\calF$, we call its $y$-coordinate the {\em height} of the point.

In the uniform case, each cdf $F_p$ has three segments: the leftmost one
is a horizontal segment with two endpoints $(-\infty,0)$ and $(x_l(p),0)$,
the middle one, whose slope is $1/(x_r(p)-x_l(p))$,
has two endpoints $(x_l(p),0)$ and $(x_r(p),1)$, and the rightmost one is a
horizontal segment with two endpoints $(x_r(p),1)$ and $(+\infty,1)$.
We transform each $F_p$ to the line $l_p$ containing the middle segment of
$F_p$. Consider an unbounded interval $I$ with $x_l=-\infty$.
We can use $l_p$ to compute $\Pr[p\in I]$
in the following way. Suppose the height of the intersection of $L$
and $l_p$ is $y$. Then, $\Pr[p\in I]=0$ if $y<0$,  $\Pr[p\in I]=y$ if
$0\leq y\leq 1$, $\Pr[p\in I]=1$ if $y>1$. Therefore, once we know
$l_p\cap L$, we can obtain $\Pr[p\in I]$
in constant time. Hence, we can use $l_p$ instead of $F_p$ to
determine the $I$-probability of $p$. The advantage of using $l_p$ is
that lines are usually easier to deal with than line segments.
Below, with a little abuse of
notation, for the uniform case
we simply use $F_p$ to denote the line $l_p$ for any $p\in
P$ and now $\calF$ is a set of lines.

Fix the query interval $I=[x_l,x_r]$.
For each $i$, $1\leq i\leq n$, denote by $p_i$ the point of $P$ whose
$I$-probability is the $i$-th largest.
Based on the above discussion, we obtain Lemma \ref{lem:20}, which holds for both
the histogram and uniform cases.

\begin{lemma}\label{lem:20}
If $x_l=-\infty$, then for each $1\leq i\leq n$, $p_i$ is
the point of $P$ such that
$L\cap F_{p_i}$ is the
$i$-th highest among the intersections of $L$ and all cdfs of $\calF$.\qed
\end{lemma}

Suppose $x_l=-\infty$. Based on Lemma \ref{lem:20}, to answer the top-1 query on $I$, it is
sufficient to find the cdf of $\calF$ whose intersection with $L$ is
the highest; to answer the top-$k$ query, it is
sufficient to find the $k$ cdfs of $\calF$ whose intersections with
$L$ are the highest; to answer the threshold query on $I$  and $\tau$,
it is sufficient to find the cdfs of $\calF$ whose intersections with
$L$ have $y$-coordinates $\geq \tau$.


\vspace{0.2cm}
\noindent
{\bf
Half-plane range reporting:}
As the half-plane range reporting data structure~\cite{ref:ChazelleTh85}
is important for our later developments,  we briefly discuss it in the dual setting.
Let $S$ be a set of $n$ lines. Given any point $q$, the goal is to report
all lines of $S$ that are above $q$. An $O(n)$-size data structure can be built in
$O(n\log n)$ time that can answer each
query in $O(\log n+m')$ time, where $m'$ is the number of lines above
the query point $q$ \cite{ref:ChazelleTh85}. The data structure can be
built as follows.

Let $U_S$ be the upper envelope of $S$
(e.g., see Fig.~\ref{fig:layers}).
We represent $U_S$ as an array of lines $l_1,l_2,\ldots,l_h$ ordered
as they appear on $U_S$ from left to right.
For each line $l_i$, $l_{i-1}$
is its {\em left neighbor} and $l_{i+1}$
is its {\em right neighbor}.
We partition $S$ into a sequence $L_1(S),L_2(S),\ldots$,
of subsets, called {\em layers} (e.g., see Fig.~\ref{fig:layers}).
The first layer $L_1(S)\subseteq S$
consists of the lines that appear on $U_S$. For $i>1$, $L_i(S)$
consists of the lines that appear on the upper envelope of the lines
in $S\setminus\bigcup_{j=1}^{i-1}L_j(S)$. Each layer $L_i(S)$ is
represented in the same way as $U_S$.
To answer a half-plane range reporting
query on a point $q$, let $l(q)$ be the vertical line
through $q$. We first determine the line $l_i$ of $L_1(S)$ whose
intersection with $l(q)$ is on the upper envelope of $L_1(S)$, by
doing binary search on the array of lines of $L_1(S)$.
Then, starting from $l_i$, we walk on the upper envelope of $L_1(S)$
in both directions to report the lines of $L_1(S)$ above the point
$q$, in linear time with respect to the output size. Next, we find the
line of $L_2(S)$ whose intersection with $l(q)$ is on the upper
envelope of $L_2(S)$. We use the same procedure as for $L_1(S)$
to report the lines of $L_2(S)$ above $q$.
Similarly, we continue on the layers $L_3(S),L_4(S),\ldots$, until no
line is reported in a certain layer. By using fractional cascading
\cite{ref:ChazelleFr86}, after determining the line $l_i$ of $L_1(S)$ in
$O(\log n)$ time by binary search, the data structure \cite{ref:ChazelleTh85} can report
all lines above $q$ in constant time each.

For any vertical line $l$,
for each layer $L_i(S)$, denote by $l_i(l)$ the line of
$L_i(S)$ whose intersection with $l$ is on the upper envelope of $L_i(S)$.
By fractional cascading \cite{ref:ChazelleFr86},
we have the following lemma for
the data structure \cite{ref:ChazelleTh85}.

\begin{lemma}\label{lem:30}{\em \cite{ref:ChazelleFr86,ref:ChazelleTh85}}
For any vertical line $l$, after the line $l_1(l)$ is known,
we can obtain the lines $l_2(l),l_3(l),\ldots$ in this order in $O(1)$ time each.\qed
\end{lemma}


\begin{figure}[t]
\begin{minipage}[t]{\linewidth}
\begin{center}
\includegraphics[totalheight=1.2in]{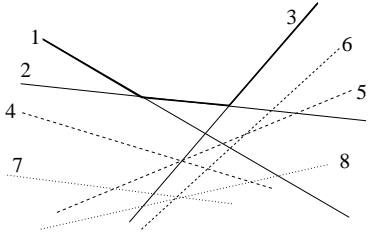}
\caption{\footnotesize Partitioning $S$ into three layers:
$L_1(S)=\{1,2,3\}$, $L_2(S)=\{4,5,6\}$, $L_3(S)=\{7,8\}$. The thick
polygonal chain is the upper envelope of $S$.}
\label{fig:layers}
\end{center}
\end{minipage}
\vspace*{-0.15in}
\end{figure}


\section{The Uniform Distribution}
\label{sec:uniform}



In this section, we present our results for the uniform case.
We first discuss our data structures for the unbounded case in Section \ref{sec:unbounded}, which will also be needed in our data structures for the bounded case in
Section \ref{sec:uniformbounded}. Further,
the results in Section \ref{sec:unbounded} will also be useful
in our data structures for the histogram case in
Section \ref{sec:nonuniform}.

Recall that in the uniform case $\calF$ is a set of lines.

\subsection{Queries with Unbounded Intervals}
\label{sec:unbounded}

We first discuss the unbounded case where $I=[x_l,x_r]$ is unbounded and some techniques introduced here will also be used later for the bounded case.
Without loss of generality, we assume $x_l=-\infty$, and the other case where $x_r=+\infty$ can be solved similarly. Recall that $L$ is the vertical line with $x$-coordinate $x_r$.

%
%

For top-1 queries, by Lemma \ref{lem:20}, we only need to maintain the upper envelope of $\calF$,  which can be computed in $O(n\log n)$ time and $O(n)$ space. For each query, it is sufficient to determine the
intersection of $L$ with the upper envelope of $\calF$, which can be done in $O(\log n)$ time.

Next, we consider top-$k$ queries.

Given $I$ and $k$,
by Lemma \ref{lem:20}, it
suffices to find the $k$ lines of $\calF$ whose
intersections with $L$ are the highest, and we let $\calF_k$ denote
the set of the above $k$ lines.
As preprocessing, we build the half-plane range reporting data structure (see Section
\ref{sec:pre}) on $\calF$, in $O(n\log n)$ time and $O(n)$ space.
Suppose the layers of $\calF$ are $L_1(\calF),L_2(\calF),\ldots$.
In the sequel, we compute the set $\calF_k$.
Let the lines in $\calF_k$ be $l^1,l^2,\ldots,l^k$ ordered from top
to bottom by their intersections with $L$.

Let $l_i(L)$ be the line of
$L_i(\calF)$ which intersects $L$ on the upper
envelope of the layer $L_i(\calF)$, for $i=1,2,\ldots$.
We first compute $l_1(L)$ in $O(\log n)$ time by
binary search on the upper envelope of $L_1(\calF)$. Clearly, $l^1$ is $l_1(L)$.
Next, we determine $l^2$. Let the set $H$ consist of the
following three lines: $l_2(L)$, the left neighbor (if any)
of $l_1(L)$ in
$L_1(\calF)$, and the right neighbor (if any) of $l_1(L)$ in $L_1(\calF)$.

\begin{lemma}\label{lem:40}
$l^2$ is the line in $H$ whose intersection with $L$ is the highest.
\end{lemma}
\begin{proof}
Note that $l^2$ is the line of $\calF\setminus\{l^1\}$
whose intersection with $L$ is the highest.
We distinguish two cases:

\begin{enumerate}
\item
If $l^2$ is in $L_1(\calF)$,
since the slopes of the
lines of $L_1(\calF)$ from left to right are increasing,
$l^2$ must be a neighbor of $l^1$.
Hence, $l^2$ must be either the left neighbor or the right neighbor of
$l^1$ in $L_1(\calF)$.

\item
If $l^2$ is not in $L_1(\calF)$, then
$l_2(L)$ must be the line of $\calF\setminus L_1(\calF)$ whose
intersection with $L$ is the highest. According to the definition of
the layers of $\calF$ the upper envelope of $L_2(\calF)$ is also the
upper envelope of $\calF\setminus L_1(\calF)$. Therefore,
$l_2(L)$ is the line of $\calF\setminus L_1(\calF)$ whose
intersection with $L$ is the highest. Hence, $l^2$ must
be $l_2(\calF)$.
\end{enumerate}

The lemma thus follows.\qed
\end{proof}

We refer to $H$ as the {\em candidate set}. By Lemma \ref{lem:40}, we find $l^2$ in $H$ in $O(1)$ time. We remove $l^2$ from $H$, and below we insert at most three lines into $H$ such that $l^3$ must be in $H$. Specifically, if $l^2$
is $l_2(L)$, we insert the following three lines into $H$:
$l_3(L)$, the left neighbor of $l_2(L)$, and the right neighbor of
$l_2(L)$. If $l^2$ is the left (resp., right) neighbor $l$ of $l_1(L)$, we insert
the left (resp., right) neighbor of $l$ in $L_1(\calF)$ into $H$.
By generalizing Lemma \ref{lem:40}, we can show $l^3$
must be in $H$ (the details are omitted). We repeat the same
algorithm until we find $l^k$. To facilitate the implementation,
we use a heap to store the lines of $H$ whose ``keys'' in the
heap are the heights of the intersections of $L$ and the lines of
$H$.

\begin{lemma}\label{lem:new50}
The set $\calF_k$ can be found in $O(\log n+k\log k)$ time.
\end{lemma}
\begin{proof}
According to our algorithm, there are
$O(k)$ insertions and ``Extract-Max'' operations (i.e., finding
the element of $H$ with the largest key and remove the element from
$H$) on the heap $H$. The size of $H$ is always bounded by
$O(k)$ during the algorithm. Hence all operations on $H$ take $O(k\log k)$ time.
Further, after finding
$l_1(L)$ in $O(\log n)$ time, due to Lemma \ref{lem:30}, the
lines that are inserted into $H$ can be found in constant time
each. Hence, the total time for finding $\calF_k$ is
$O(k\log k+\log n)$. \qed
\end{proof}

We can improve the algorithm to $O(\log n+k)$ time by using the selection algorithm in \cite{ref:FredericksonTh82} for sorted arrays. The key idea is that we can implicitly obtain $2k$ sorted arrays of $O(k)$ size each and $\calF_k$ can be computed by finding the largest $k$ elements in these arrays.
The details are given in Lemma \ref{lem:new60}.

\begin{lemma}\label{lem:new60}
The set $\calF_k$ can be found in $O(\log n+k)$ time.
\end{lemma}
\begin{proof}
Consider any layer $L_i(\calF)$.
Suppose the array of lines of $L_i(\calF)$ is $l_1,l_2,\ldots,l_h$ and let $l_j$ be the line $l_i(L)$. The intersections of the lines $l_j,l_{j+1},\ldots,\l_h$ with $L$ are sorted in decreasing order of their heights, and the intersections of the lines $l_{j-1},l_{j-2},\ldots,\l_1$ with $L$ are also sorted in decreasing order of their heights. Once $l_j$ is known, we can implicitly obtain the following two arrays $A^r_{i}$ and $A^l_{i}$: the $t$-th element of $A^r_{i}$ (resp., $A^l_{i}$) is the height of the intersection of $l_{t-j+1}$ (resp., $l_{j-t}$) and $L$. Since these lines are explicitly maintained in the layer $L_i(\calF)$, given any index $t$, we can obtain the $t$-th element of $A^r_{i}$ (resp., $A^l_{i}$) in $O(1)$ time.

To compute the set $\calF_k$, we first find the lines $l_i(\calF)$ for $i=1,2,\ldots,k$, which can be done in $O(\log n)$ time due to Lemma \ref{lem:30}.  Consequently, we obtain the $2k$ arrays $A_i^r$ and $A_i^l$ for $1\leq i\leq k$, implicitly. In fact we only need to consider the first $k$ elements of each such array, and below we let $A_i^r$ and $A_i^l$ denote the arrays only consisting of the first $k$ elements. An easy observation is that the heights of the intersections of $L$ and the lines of $\calF_k$ are exactly the largest $k$ elements of $\calA=\bigcup_{i=1}^k\{A_i^r\cup A_i^l\}$.

In light of the above discussion, to compute $\calF_k$, we do the following: (1) find the $k$-th largest element $\tau$ of $\calA$; (2) find the lines of $\calA$ whose intersections with $L$ have heights at least $\tau$, which can be done in $O(k)$ time by checking the above $2k$ sorted arrays with $\tau$ in their index orders. Below, we show that we can compute $\tau$ in $O(k)$ time.

Recall that $\calA$ contains $2k$ sorted arrays and each array has $k$ elements. Further, for any array, given any index $t$, we can obtain its $t$-th element in constant time.
Hence, we can find the $k$-th largest element of $\calA$ in $O(k)$ time by using the selection algorithm given in \cite{ref:FredericksonTh82} for matrices with sorted columns (each sorted array in our problem can be viewed as a sorted column of a $k\times 2k$ matrix).

The lemma thus follows.\qed
\end{proof}

Hence, we obtain the following result.

\begin{theorem}\label{theo:uniformunboundedtopk}
For the uniform case,
we can build in $O(n\log n)$ time an $O(n)$ size
data structure on $P$ that can answer each top-$k$ query with an unbounded query
interval in $O(k+\log n)$ time.
\end{theorem}

For the threshold query, we are given $I$ and a threshold $\tau$.
We again build the half-plane range reporting data structure on $\calF$. To answer the query, as discussed in Section \ref{sec:pre}, we only need to
find all lines of $\calF$ whose intersections with $L$ have
$y$-coordinates larger than or equal to $\tau$.
We first determine the line $l_1(L)$ by doing binary search on the
upper envelope of $L_1(\calF)$. Then, by Lemma \ref{lem:30}, we
find all lines $l_2(L),l_3(L),\ldots,l_j(L)$ whose intersections have
$y$-coordinates larger than or equal to $\tau$.
For each $i$ with $1\leq i\leq j$, we walk on the upper envelope
of $L_i(\calF)$, starting from $l_i(L)$, on both directions in
time linear to the output size to find the lines whose intersections
have $y$-coordinates larger than or equal to $\tau$.
Hence, the running time for
answering the query is $O(\log n+m)$.


\subsection{Queries with Bounded Intervals}
\label{sec:uniformbounded}

Now we assume $I=[x_l,x_r]$ is bounded.
Consider any point $p\in P$. Recall that $p$ is associated with an interval
$[x_{l}(p),x_{r}(p)]$ in the uniform case. Depending on the positions of $I=[x_l,x_r]$ and $[x_{l}(p),x_{r}(p)]$, we classify $[x_l(p),x_r(p)]$ and the point
$p$ into the following three types with respect to $I$.


\begin{description}
\item[L-type:] $[x_l(p),x_r(p)]$ and $p$ are L-type if $x_l\leq
x_l(p)$.
\item[R-type:] $[x_l(p),x_r(p)]$ and $p$ are R-type if $x_r\geq
x_r(p)$.
\item[M-type:] $[x_l(p),x_r(p)]$ and $p$ are M-type if $I\subset
(x_l(p),x_r(p))$.
\end{description}


Denote by $P_L$, $P_R$, and $P_M$ the sets of all $L$-type, $R$-type,
and $M$-type of points of $P$, respectively.
In the following, for each kind of query, we will build an data structure such
that the different types of points will be searched separately (note
that we will not explicitly compute the three subsets $P_L$, $P_R$,
and $P_M$).
For each point $p\in P$, we refer to $x_l(p)$ as the {\em left
endpoint} of the interval $[x_l(p),x_r(p)]$ and refer to $x_r(p)$
as the {\em right endpoint}. For simplicity of discussion, we assume that no two interval endpoints of the points of $P$ have the same value.

\subsubsection{Top-1 Queries}

For any point $p\in P$, denote by $\calF_r(p)$ the set of the cdfs
of the points of $P$ whose intervals have
left endpoints larger than or equal to $x_l(p)$. Again, as
discussed in Section \ref{sec:pre} we
transform each cdf of $\calF_r(p)$ to a line. We aim to maintain
the upper envelope of $\calF_r(p)$ for each $p\in P$.
If we computed the $n$ upper envelopes
explicitly, we would have an data structure of size $\Omega(n^2)$. To
reduce the space, we choose to use the persistent data
structure \cite{ref:DriscollMa89} to maintain
them implicitly such that data structure size is $O(n)$. The details
are given below.

We sort the points of $P$ by the left endpoints of
their intervals from left to right, and let the sorted list be
$p'_1,p'_2,\ldots,p'_n$. For each $i$ with $2\leq i\leq n$, observe that
the set $\calF_r(p'_{i-1})$ has exactly one more line
than $\calF_r(p'_i)$. If we maintain the upper
envelope of $\calF_r(p'_i)$ by a balanced binary search tree
(e.g., a red-black tree), then by updating it we
can obtain the upper envelope of $\calF_r(p'_{i-1})$ by an insertion
and a number of deletions on the tree, and each tree operation takes
$O(\log n)$ time. An easy observation is that there are $O(n)$
tree operations in total to compute the upper envelopes of all sets
$\calF_r(p'_1),\calF_r(p'_2),\ldots,\calF_r(p'_n)$.
Further, by making the red-black tree persistent \cite{ref:DriscollMa89}, we
can maintain all upper envelopes in $O(n\log n)$ time
and $O(n)$ space.  We use $\calL$ to denote the above data structure.

We can use $\calL$ to find the point of $P_L$ with the largest
$I$-probability in $O(\log n)$ time, as follows.
First, we find the
point $p'_i$ such that $x_l(p'_{i-1})<x_l\leq x_l(p'_i)$.
It is easy to see that $\calF_r(p'_i)=P_L$.
Consider the unbounded interval $I'=(-\infty,x_r]$.
Consider any point $p$ whose cdf is in $\calF_r(p'_i)$. Due to
$x_l(p)\geq x_l$, we can obtain that $\Pr[p\in I]=\Pr[p\in I']$. Hence,
the point $p$ of $\calF_r(p'_i)$ with the largest value
$\Pr[p\in I]$ also has the largest value $\Pr[p\in I']$. This implies
that we can instead use the unbounded interval $I'$ as the query
interval on the upper envelope of $\calF_r(p'_i)$, in the same way as in Section \ref{sec:unbounded}.
The persistent data structure $\calL$
maintains the upper envelope of $\calF_r(p'_i)$ such that we can find in
$O(\log n)$ time the
point $p$ of $\calF_r(p'_i)$ with the largest value $\Pr[p\in I']$.

Similarly, we can build a data structure $\calR$ of $O(n)$
space in $O(n\log n)$ time that can find
the point of $P_R$ with the largest $I$-probability in $O(\log n)$
time.

To find the point of $P_M$ with
the largest $I$-probability, the approach for $P_L$ and
$P_R$ does not work because we cannot reduce the query to
another query with
an unbounded interval. Instead, we reduce the problem to a ``segment
dragging query'' by dragging a line segment out of a corner in the
plane, as follows.

For each point $p$ of $P$, we define a point $q=(x_l(p),x_r(p))$ in the
plane, and we say that $p$ {\em corresponds to} $q$. Similar transformation was also used in \cite{ref:ChengEf04}. Let $Q$ be the set of the $n$ points defined by the points of $P$. For the query interval $I=[x_l,x_r]$, we also define a point
$q_I=(x_l,x_r)$ (this is different from \cite{ref:ChengEf04}, where
$I$ defines a point $(x_r,x_l)$). If we partition the plane into
four quadrants with respect to $q_I$, then we have the following lemma.

\begin{lemma}\label{obser:10}
The points of $P_M$ correspond to the points of $Q$ that
strictly lie in the second quadrant (i.e., the northwest quadrant) of $q_I$.
\end{lemma}
\begin{proof}
Consider any point $p\in P$. Let $q=(x_l(p),x_r(p))$ be the point
defined by $p$. On the one hand, $p$ is in $P_M$ if and only if $I\subset(x_l(p),x_r(p))$,
i.e., $x_l>x_l(p)$ and $x_r<x_r(p)$. On the other hand,
$x_l>x_l(p)$ and $x_r<x_r(p)$
if and only if $q$ is in the second quarter
of $q_I=(x_l,x_r)$. The lemma thus follows. \qed
\end{proof}

Let $\rho_u$ be the upwards ray
originating from $q_I$ and let $\rho_l$ be the leftwards ray originating
from $q_I$. Imagine that starting
from the point $q_I$ and towards northwest,
we drag a segment of slope $1$ with two
endpoints on $\rho_u$ and $\rho_l$ respectively, and let $q^*$ be the
point of $Q$ hit first by the segment (e.g., see Fig.~\ref{fig:segdrag}).

\begin{figure}[t]
\begin{minipage}[t]{\linewidth}
\begin{center}
\includegraphics[totalheight=1.2in]{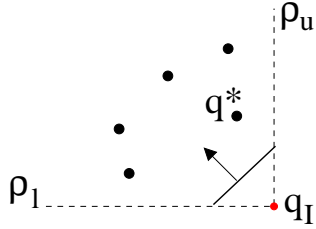}
\caption{\footnotesize Dragging a segment of slope $1$
out of the corner at $q_I$: $q^*$ is the first point
that will be hit by the segment.}
\label{fig:segdrag}
\end{center}
\end{minipage}
\vspace*{-0.15in}
\end{figure}

\begin{lemma}\label{lem:50}
The point of $P$ that defines $q^*$ is in $P_M$ and
has the largest $I$-probability
among all points in $P_M$.
\end{lemma}
\begin{proof}
First of all, by Lemma \ref{obser:10}, $q^*$ must be in $P_M$.

Consider any point $q$ in the second quadrant of $q_I$, and let
$p$ be the point of $P$ that defines $q$. Since the interval of $p$
contains the interval $I$, we have $\Pr[p\in
I]=\frac{x_r-x_l}{x_r(p)-x_l(p)}$.

Based on the definition of $q^*$, $q^*$ is the point $q$ of $Q$ in the
second quadrant of $q_I$ that has the smallest value
$x_r(p)-x_l(p)$. Therefore,  $q^*$ is the point $q$ of $Q$ in the
second quadrant of $q_I$ that has the largest value
$\frac{x_r-x_l}{x_r(p)-x_l(p)}$. The
lemma thus follows. \qed
\end{proof}

Based on Lemma \ref{lem:50}, to determine the point of $P_M$ with the
largest $I$-probability, we only need to solve the above query on
$Q$ by dragging a segment out of a corner. More specifically, we need
to build a data structure on $Q$ to answer the following {\em
out-of-corner segment-dragging queries}: Given a point $q$, find the
first point of $Q$ hit by dragging a segment of slope $1$ from $q$ and
towards the northwest direction with the two endpoints on the two rays
$\rho_u(q)$ and $\rho_l(q)$, respectively, where $\rho_u(q)$ is the
upwards ray originating from $q$ and $\rho_l(q)$ is the leftwards ray
originating from $q$.  By using Mitchell's result in
\cite{ref:MitchellL192} (reducing the problem to a point location
problem), we can build an $O(n)$
size data structure on $Q$ in $O(n\log n)$ time that can answer each
such query in $O(\log n)$  time.

\begin{theorem}
For the uniform case, we can build in $O(n\log n)$ time an $O(n)$ size
data structure on $P$ that can answer each top-$1$ query in $O(\log n)$ time. \qed
\end{theorem}

\subsubsection{Top-k Queries}
\label{sec:uniformtopk}
To answer a top-k query, we will do the following. First, we find the top-$k$ points in $P_L$ (i.e., the $k$ points of $P_L$ whose
$I$-probabilities are the largest), the top-$k$ points in $P_R$, and the
top-$k$ points in $P_M$. Then, we find the top-$k$ points of $P$
from the above $3k$ points.
Below we build three data structures for computing the
top-$k$ points in $P_L$, $P_R$, and $P_M$, respectively.

We first build the data structure for
$P_L$. Again, let $p'_1,p'_2,\ldots,p'_n$ be the list of the points
of $P$ sorted by the left endpoints of their intervals from left to
right. We construct a complete binary search tree $T_L$ whose
leaves from left to right store the $n$ intervals of the points
$p'_1,p'_2,\ldots,p'_n$. For each internal node $v$, let $P_v$ denote the
set of points whose intervals are stored in the leaves of the subtree
rooted at $v$. We build the half-plane range reporting
data structure discussed in Section \ref{sec:pre} on $P_v$, denoted by $D_v$.
Since the size of $D_v$ is $|P_v|$, the total
size of the data structure $T_L$ is $O(n\log n)$, and $T_L$ can be
built in $O(n\log^2 n)$ time.

We use $T_L$ to compute the top-$k$ points in $P_L$ as follows. By the standard approach and using $x_l$, we find in $O(\log n)$ time a set $V$ of $O(\log
n)$ nodes of $T_L$ such that $P_L=\bigcup_{v\in V}P_v$ and no node of $V$ is an ancestor of another node.
Then, we can determine the top-$k$ points of $P_L$ in similarly
as in Section \ref{sec:unbounded}.
However, since we now have $O(\log n)$  data structures $D_v$, we need
to maintain the candidate sets for all such $D_v$'s. Specifically,
after we find the top-1 point
in $D_v$ for each $v\in V$, we use a heap $H$ to maintain them where
the ``keys'' are the $I$-probabilities of the points.
Let $p$ be the point of $H$ with the largest key.
 Clearly, $p$ is the top-1 point of $P_L$;
assume $p$ is from $D_v$ for some $v\in V$.
We remove $p$ from $H$ and insert at most three new points
from $D_v$ into $H$, in a similar way as in
Section \ref{sec:unbounded}. We repeat the same procedure until we
find all top-$k$ points of $P_L$.

To analyze the running time, for each node $v\in V$, we can
determine in $O(\log n)$ time
the line in the first layer of $D_v$ whose intersection with $L$ is
on the upper envelope of the first layer,
and subsequent operations on $D_v$ each takes $O(1)$ time due to
fractional cascading. Hence, the total time for this step in the entire algorithm is
$O(\log^2 n)$. However, we can do better by building a fractional cascading
structure \cite{ref:ChazelleFr86} on the first layers of $D_v$ for all
nodes $v$ of the tree $T_L$. In this way, the above step only takes
$O(\log n)$ time in the entire algorithm,
i.e., do binary search only at the root of $T_L$.
In addition, building the heap $H$ initially takes $O(\log n)$ time. Note that the additional fractional cascading structure on $T_L$ does not change the size and construction time of $T_L$ asymptotically \cite{ref:ChazelleFr86}.
The entire query algorithm has $O(k)$ operations
on $H$ in total and the size of $H$ is $O(\log n+k)$. Hence, the total
time for finding the top-$k$ points of $P_L$ is
$O(\log n+k\log(k+\log n))$, which is
$O(\log n+k\log k)$ by Lemma \ref{lem:time}.

\begin{lemma}\label{lem:time}
$\log n+k\log(k+\log n)=O(\log n+k\log k)$.
\end{lemma}
\begin{proof}
To simplify the notation, let $n'=\log n$, and our goal is to prove
$k\log(k+n')+n'=O(k\log k+n')$. Depending on whether $k\geq
\frac{n'}{\log n'}$, there are two cases.

\begin{enumerate}
\item
If $k\geq \frac{n'}{\log n'}$, then $\log k\geq \log n'-\log\log n'$,
implying that $\log\log n'=O(\log k)$. Thus,
\begin{equation*}
\begin{split}
& k\log (k+n')\leq k\log(k+k\log n') = O(k\log (k\log
n'))\\
&=O(k(\log k+\log\log n'))=O(k\log k).
\end{split}
\end{equation*}

Hence, we obtain that $k\log(k+n')+n'=O(k\log k+n')$.

\item
If $k<\frac{n'}{\log n'}$, then $k\log (k+n')\leq \frac{n'}{\log
n'}\log (\frac{n'}{\log n'}+n')=O(\frac{n'}{\log
n'}\log n')=O(n')$.

Hence, we obtain that $k\log(k+n')+n'=O(k\log k+n')$.
\end{enumerate}

The lemma thus follows. \qed
\end{proof}

If $k=\Omega(\log n\log\log n)$, we have a better result in Lemma \ref{lem:100}. Note that comparing with Lemma \ref{lem:new60}, we need to use other techniques to obtain Lemma \ref{lem:100} since the problem here involves $O(\log n)$ half-plane range reporting data structures $D_v$ while Lemma \ref{lem:new60} only needs to deal with one such data structure.

\begin{lemma}\label{lem:100}
If $k=\Omega(\log n\log\log n)$, we can compute the top-$k$ points in $P_L$ in $O(k)$ time.
\end{lemma}
\begin{proof}
We assume $k=\Omega(\log n\log\log n)$.
Recall that $L$ is the vertical line with $x$-coordinate $x_r$.
Let $V$ be the set of $O(\log n)$ nodes of $T_L$ as defined before. Consider any node $v\in V$, which is associated with a half-plane range reporting data structure $D_v$ on the cdfs of the point set $P_v$. Let $\calF(v)$ be the cdfs of the points of $P_v$. Let $\calF(V)=\cup_{v\in V}\calF(v)$.  Our goal is to find the $k$ lines of $\calF(V)$ whose intersections with $L$ are the highest, and denote by $\calF_k$ the above $k$ lines that we seek.

We let $L_1(v),L_2(v),\ldots$ be the layers of $\calF(v)$, and for each layer $L_i(v)$, denote by $l_i(v)$ the line of $L_i(v)$ whose intersection with $L$ is on the upper envelope of the layer $L_i(v)$. For each layer $L_i(v)$, we define two arrays $A^r_i(v)$ and $A^l_i(v)$ in the same way as in the proof of Lemma \ref{lem:new60} (we omit the details). For each node $v$, we define another array $B(v)$ of size $k$ as follows: for each $1\leq i\leq k$, the $i$-th element of $B(v)$ is the height of the intersection of $l_i(v)$ and $L$. Hence, the elements of $B(v)$ are sorted in decreasing order.

Our algorithm for computing $\calF_k$ has two main steps. In the first main step, we will find a set $B'$ of the  largest $k$ elements in $B(V)=\cup_{v\in V}B(v)$. For each $v\in V$, let $j_v$ be the number of  elements of $B(v)$ that are contained in $B'$, i.e., the first $j_v$ elements of $B(v)$ are in $B'$. An easy observation is that the heights of the intersections of $L$ and the sought lines of $\calF_k$ are the largest $k$ elements of $\calA=\bigcup_{v\in V}\cup_{i=1}^{j_v}\{A_i^r(v)\cup A_i^l(v)\}$, which contains $2k$ sorted arrays. In the second main step, we will find the largest $k$ elements of $\calA$ and thus obtain the set $\calF_k$. Below, we show that the above two main steps can be done in $O(k)$ time.

We consider the first main step. For simplicity of discussion, we assume no two elements of $B(V)$ are equal.
First of all, as discussed above, in $O(\log n)$ time we can determine the lines $l_1(v)$ for all $v\in V$, and thus the first elements of all arrays $B(v)$ for $v\in V$ are determined. Further, for each array $B(v)$, due to the fractional cascading, we can obtain the next element in constant time each, and in other words, we can obtain the first $i$ elements of $B(v)$ in $O(i)$ time. To compute the set $B'$, i.e., the largest $k$  elements of $B(V)$,
since $B(V)$ contains $|V|$ sorted arrays, one may want to use the selection algorithm in \cite{ref:FredericksonTh82} again, as in Lemma \ref{lem:new60}. However, here we cannot use that algorithm because for each array in $B(V)$, given any data structure, we cannot obtain the corresponding element  in constant time. Instead, we propose the following approach.

Recall that $k=\Omega(\log n\log \log n)$. For simplicity of discussion, we assume $|V|=\log n$ and $k\geq \log n\log \log n$. Let $h=\log\log n$.
Let $H$ be a max-heap. Initially $H=\emptyset$.  During the algorithm, we will maintain a set $S$ of elements. Initially $S=\emptyset$, and after the algorithm stops, $S$ contains at most $k+h$ elements.


First of all, for each array $B(v)$, we compute its first $h$ elements and then insert only the $h$-th element of $B(v)$ into $H$. Now $H$ contains $|V|=\log n$ elements. We do an ``extract-max'' operation on $H$, i.e., remove the largest element from $H$. Suppose the element removed above is from $B(v)$ for a node $v$. Then, we add the first $h$ elements of $B(v)$ into $S$. If $|S|\geq k$, the algorithm stops; otherwise, we compute the next $h$ elements of $B(v)$ and insert only the $(2h)$-th element of $B(v)$ into $H$.

In general, suppose we do an ``extract-max'' operation on $H$ and let the removed element by the operation be the $(i\cdot h)$-th element of the array $B(v)$ for a node $v$. Then, we add the elements of $B(v)$ with indices from $i\cdot h-1$ to $i\cdot h$ to $S$. If $|S|\geq k$, the algorithm stops; otherwise, we compute the next $h$ elements of $B(v)$ and insert the $[i\cdot (h+1)]$-th element of $B(v)$ into $H$.

After the algorithm stops, we do the following. Consider any element in the current heap $H$, and suppose it is the $(j\cdot h)$-th element of the array $B(v)$ for a node $v$. Then, according to our algorithm, the elements of $B(v)$ with indices from $j\cdot h-1$ to $j\cdot h$ are not in $S$, and we call these $h$ elements the {\em red elements}. Since the current $H$ has at most $|V|-1$ elements, there are at most $h\cdot (|V|-1)$ red elements, and we let $S'$ be the union of $S$ and all red elements.

We claim that the largest $k$ elements of $B(V)$ must be in $S'$, i.e., $B'\subseteq S'$. We prove the claim as follows. Let $a$ be the element that is removed from $H$ in the last extract-max operation on $H$ in the above algorithm, i.e., after $a$ is removed, the algorithm stops. Let $H$ be the heap after the algorithm stops. According to our algorithm, since each array $B(v)$ is sorted decreasingly, $a$ is the smallest element in $S$. Since $|S|\geq k$ and $|B'|=k$, any element of $B'$ must be larger than $a$. If $B'\subseteq S$, then the claim is proved. Otherwise, suppose an element $b$ is in $B'\setminus S$. Since $b>a$ and all elements of $H$ are smaller than $a$, $b$ must be a red element, and thus $b$ is in $S'$. The claim is proved.

In light of the above claim, $B'$ can be easily obtained after we find the $k$-th largest element of $S'$.

In the sequel, we analyze the running time of the above algorithm for the first main step. Recall that $h=\log\log n$. First of all, the size of the heap $H$ is at most $|V|=\log n$ at any time during the algorithm. Clearly, the algorithm will stop after $\lceil\frac{k}{h}\rceil$ extract-max operations. The number of insertions is at most $\lceil\frac{k}{h}\rceil+|V|$. Therefore, the running time of all operations on $H$ in the entire algorithm is $O((\lceil\frac{k}{h}\rceil+\log n)\log\log n)$, which is $O(k)$ due to $k=\Omega(\log n\log \log n)$. On the other hand, the number of red elements is at most $h\cdot (\log n-1)$, and thus $|S'|\leq h\cdot (\log n-1)+|S|\leq h\cdot (\log n-1)+k+h$, which is $O(k)$ due to $k=\Omega(\log n\log \log n)$. Notice that the elements of $B(V)$ that have been computed during the entire algorithm are exactly those in $S'$, and thus the time for computing these elements is $O(|S'|)=O(k)$. Finally, since $|S'|=O(k)$, we can find the $k$-th largest element in $S'$ in $O(k)$ time by using the well-known linear time selection algorithm.

As a summary, the first main step can find the set $B'$ in $O(k)$ time.

The second main step is to compute the largest $k$ elements in $\calA$, which contains $2k$ arrays. As in the proof of Lemma \ref{lem:new60}, we can obtain any arbitrary element of these arrays in constant time, without computing these arrays explicitly. Hence, by using the same approach as in Lemma \ref{lem:new60}, we can compute the largest $k$ elements of $\calA$ in $O(k)$ time. Consequently, the set $\calF_k$ can be obtained.

The lemma thus follows.\qed
\end{proof}

To compute the top-$k$ points of $P_R$, we build a similar data
structure $T_R$, in a symmetric way as $T_L$, and we omit the details.

Finally, to compute the top-$k$ points
in $P_M$, we do the following transformation. For each point
$p\in P$, we define a point $q=(x_l(p),x_r(p),1/(x_r(p)-x_l(p))$ in
the 3-D space with $x$-, $y$-, and $z$-axes. Let $Q$ be the set of all
points in the 3-D space thus defined.
Let the query interval $I$ define an unbounded query
box (or 3D rectangle)
$B_I=(-\infty,x_l)\times(x_r,+\infty)\times(-\infty,+\infty)$.
Similar to Lemma \ref{obser:10} in Section
\ref{sec:unbounded}, the points of $P_M$ correspond exactly to
the points of $Q\cap B_I$.
Further, the top-$k$ points of $P_M$ correspond to the $k$ points of
$Q\cap B_I$ whose $z$-coordinates are the largest. Denote by $Q_I$ the
$k$ points of $Q\cap B_I$ whose $z$-coordinates are the largest.
Below we build a data structure on $Q$ for computing
the set $Q_I$ for any query interval $I$ and thus finding
the top-$k$ points of $P_M$.

We build a complete binary
search tree $T_M$ whose leaves from left to right store all points of
$Q$ ordered by the increasing $x$-coordinate.
For each internal node $v$ of $T_M$, we build an
auxiliary data structure $D_v$ as follows. Let $Q_v$ be the set of the
points of $Q$ stored in the leaves of the subtree of $T_M$
rooted at $v$. Suppose all points of $Q_v$ have $x$-coordinates less than $x_l$. Let $Q_v'$ be the points of $Q_v$ whose $y$-coordinates
are larger than $x_r$.
The purpose of the auxiliary data structure
$D_v$ is to report the points of $Q'_v$ in the decreasing $z$-coordinate order
 in constant time each after the point of $q_v$ is found,
where $q_v$ is the point of $Q_v'$ with the largest $z$-coordinate.
To achieve this goal, we use the data structure given by Chazelle and Guibas
\cite{ref:ChazelleFr862} (the one for Subproblem P1 in Section 5),
and the data structure is a {\em hive graph}
\cite{ref:ChazelleFi86}, which can be viewed as the preliminary
version of the fractional cascading techniques
\cite{ref:ChazelleFr86}. By using the result in
\cite{ref:ChazelleFr862}, we can build such a data structure $D_v$
of size $O(|Q_v|)$ in $O(|Q_v|\log |Q_v|)$ time that can first compute
$q_v$ in $O(\log|Q_v|)$ time and then report other points of $Q'_v$
in the decreasing $z$-coordinate order in constant time each.
Since the size of $D_v$ is $|Q_v|$, the
size of the tree $T_M$ is $O(n\log n)$, and $T_M$ can be
built in $O(n\log^2 n)$ time.


Using $T_M$, we find the set $Q_I$ as follows.
We first determine the set $V$ of $O(\log n)$ nodes of $T_M$ such that
$\bigcup_{v\in V}Q_v$ consists of all points of $Q$ whose
$x$-coordinates less than $x_l$ and no point of $V$ is an ancestor of another point of $V$.
Then, for each node $v\in V$, by using $D_v$, we find $q_v$, i.e., the point of
$Q_v$ with the largest $z$-coordinate, and
insert $q_v$ into a heap $H$, where the key of each point is its
$z$-coordinate.
We find the point in $H$ with the largest key and remove
it from $H$; denote the above point by $q_1'$.
Clearly, $q'_1$ is the point of $Q_I$ with the largest $z$-coordinate.
Suppose $q'_1$ is in a node $v\in V$. We proceed on $D_v$ to find the
point of $Q_v$ with the second largest $z$-coordinate and insert it into $H$.
Now the point of $H$ with the largest key is the point of $Q_I$
with the second largest $z$-coordinate. We repeat the above procedure until we
find all $k$ points of $Q_I$.

To analyze the query time, finding the set $V$ takes $O(\log n)$
time. For each node $v\in V$, the search for $q_v$ on
$D_v$ takes $O(\log n)$ time plus the time linear to the number of
points of $D_v$ in $Q_I$. Hence, the total time for searching $q_v$
for all vertices $v\in V$ is $O(\log^2 n)$ time. Similarly as before,
we can remove a logarithmic factor by building a fractional
cascading structure on the nodes of $T_M$ for searching such points
$q_v$'s, in exactly the same way as in \cite{ref:ChazelleFi86}.
With the help of the fractional cascading structure, all these
$q_v$'s for $v\in V$ can be found in $O(\log n)$ time. Note that building the fractional cascading structure does not change the
construction time and the size of $T_M$
asymptotically \cite{ref:ChazelleFi86}.
In addition, building the heap $H$ initially takes $O(\log n)$ time.
In the entire algorithm there are $O(k)$ operations on $H$ in total
and the size of $H$ is always bounded by $O(k+\log n)$. Therefore, the
running time  of the query algorithm is
$O(\log n+k\log(k+\log n))$, which is $O(\log n+k\log k)$ by Lemma
\ref{lem:time}.

Using similar techniques as in Lemma \ref{lem:100}, we obtain the following result.

\begin{lemma}\label{lem:110}
If $k=\Omega(\log n\log \log n)$, we can compute the top-$k$ points in $P_M$ in $O(k)$ time.
\end{lemma}
\begin{proof}
Consider any point $v$ in $V$, which is associated with a set $Q_v$ and a data structure $D_v$. Define an array $B(v)$ of size $k$ as follows: for each $1\leq i\leq k$, the $i$-th element of $B(v)$ is the $i$-th largest $z$-coordinate of the points of $Q_v$. As discussed above, in $O(\log n)$ time we can obtain the first elements of $B(v)$ for all $v\in V$, and after that, we can obtain the next element of each array $B(v)$ in constant time each by using the data structure $D_v$. Our goal is to find the point set $Q_I$.

Let $B(V)=\cup_{v\in V}B(v)$.
An easy observation is that the $z$-coordinates of the points of $Q_I$ are exactly the  largest $k$ elements in $B(V)$. Since $k=\Omega(\log n\log \log n)$, computing the  largest $k$ elements of $B(V)$ can be done in $O(k)$ time in the same way as the first
main step of the algorithm in the proof of Lemma \ref{lem:100}, and we omit the details.

The lemma is thus proved.
\qed
\end{proof}

We summarize our results for the top-$k$ queries below.

\begin{theorem}\label{theo:uniformtopk}
For the uniform case, we can build in $O(n\log^2 n)$ time
an $O(n\log n)$ size data structure on $P$ that can answer each top-$k$ query
in $O(k)$ time if $k=\Omega(\log n\log\log n)$ and $O(k\log k+\log n)$ time otherwise.\qed
\end{theorem}

\subsubsection{Threshold Queries}
\label{sec:uniformthreshold}

To answer the threshold queries, we build the same data structure as in
Theorem \ref{theo:uniformtopk}, i.e., the three trees $T_L$, $T_M$, and $T_R$.
The tree $T_L$ is used for finding the points $p$ of $P_L$ with $\Pr[p\in
I]\geq \tau$; $T_R$ is for finding the points $p$ of $P_R$ with $\Pr[p\in
I]\geq \tau$; $T_M$ is for finding the points $p$ of $P_M$ with $\Pr[p\in
I]\geq \tau$. The three trees $T_L$, $T_R$, and $T_M$ are
exactly the same as those for Theorem \ref{theo:uniformtopk}. We can compute them in $O(n\log^2 n)$ time and $O(n\log n)$ space.

Below, we discuss the query algorithms on the three trees.
Let $m_L$, $m_R$, and $m_M$ be the number of points in $P_L$, $P_R$,
and $P_M$ whose $I$-probabilities are at least $\tau$, respectively.
Hence, $m=m_L+m_R+m_M$.

To find the points $p$ of $P_L$ with $\Pr[p\in
I]\geq \tau$, we first determine the set $V$ of $O(\log n)$
nodes of $T_L$ such that $\bigcup_{v\in V}P_v=P_L$ and no node of $V$ is an ancestor of another node of $V$. Recall that each node $v$ of $T_L$ is associated with a half-plane range reporting data structure
$D_v$.  For each node $v\in V$, by using $D_v$, we can find the points $p$
of $P_v$ with $\Pr[p\in I]\geq \tau$ in
$O(\log n+m_v)$ time, where $m_v$ is the output size.
Note that $P_v$ and $P_u$ are disjoint for any two nodes $v$ and $u$
of $V$. Hence, $\sum_{v\in V}m_v=m_L$.
As there are $O(\log n)$ nodes in $V$, it takes $O(\log^2 n+m_L)$ time
to find all points $p$ of $P_L$ with $\Pr[p\in I]\geq \tau$, and again the
$O(\log^2 n)$ time factor can be reduced to $O(\log n)$ by using
fractional cascading \cite{ref:ChazelleFr86}. Hence, the total query time is $O(\log
n+m_L)$.

We can use the similar approach to find all points $p$ of $P_R$
with $\Pr[p\in I]\geq \tau$ in $O(\log n +m_R)$ time,
by using $T_R$.  We omit the details.

Finally, we find the points $p$ of $P_M$ with $\Pr[p\in I]\geq \tau$,
by using $T_M$. As in Section \ref{sec:uniformtopk}, for each point
$p\in P$, we define in the 3-D space a point $(x_l(p),x_r(p),1/(x_r(p)-x_l(p))$.
Let $Q$ be the set of all $n$ points defined
above.  Let the interval $I=[x_l,x_r]$ and $\tau$ together define
an unbounded 3-D box query $B_I=(-\infty,x_l)\times(x_r,+\infty)\times[\tau,+\infty)$.
Let $Q_I=Q\cap B_I$. Hence, the points $p$ of $P_M$ with $\Pr[p\in I]\geq \tau$ correspond to the
points of $Q_I$, and thus $m_M=|Q_I|$.

By using the tree $T_M$, we can find $Q_I$ in $O(\log n+m_M)$ time,
as follows. We first determine the set $V$ of $O(\log
n)$ nodes of $T_M$ such that $\bigcup_{v\in V}Q_v$ consists of all points
of $Q$ whose $x$-coordinates are less than $x_l$ and no node of $V$ is an ancestor of another node of $V$. Consider any node $v\in V$. Let $Q_v'$ be the points of $Q_v$ whose
$y$-coordinates are larger than $x_r$, and $q_v$ be the point $Q_v'$ with the largest $z$-coordinate.  Recall that after $q_v$ is found $D_v$ can report other points of $Q'_v$ in the decreasing $z$-coordinate order in constant time each. Hence after $q_v$ is known we can report the points of $Q_v$ in the query box $B_I$ in time linear to the output size. Again, with the help of fractional cascading, the nodes $q_v$ for all $v\in V$ can be found in $O(\log n)$ time.
Therefore, we can find all points of $Q_I$ in $O(\log n+m_M)$ time.
In other words, with $T_M$, we can find the points of
of $P_M$ whose $I$-probabilities at least $\tau$ in $O(\log n+m_M)$ time.

Hence, we obtain Theorem \ref{theo:uniformthreshold}.

\begin{theorem}\label{theo:uniformthreshold}
For the uniform case, we can build in $O(n\log^2 n)$ time
an $O(n\log n)$ size
data structure that can answer each threshold query
in $O(m+\log n)$ time, where $m$ is the output size of the query.
\end{theorem}

\section{The Histogram Distribution}
\label{sec:nonuniform}

In this section, we present our data structures for the histogram case.
In the histogram case, the cdf of each point $p\in P$ has $c$
pieces; recall that we assumed $c$ is a constant, and thus
$\calF$ is still a set of $O(n)$ line segments.

We first discuss our data structures for the unbounded case in Section \ref{sec:unboundedhist} and then present our results for the bounded case in Section \ref{sec:boundedhist}.

\subsection{Queries with Unbounded Intervals}
\label{sec:unboundedhist}

Again, we assume w.l.o.g. that $x_l=-\infty$. Recall that
$L$ is the vertical line with $x$-coordinate $x_r$.
Note that Lemmas \ref{lem:10} and \ref{lem:20} are still applicable.

\subsubsection{Top-1 Queries}
For the top-1 queries, as in Section \ref{sec:unbounded} it is
sufficient to maintain the upper envelope of $\calF$. Although
$\calF$ now is a set of line segments, its upper envelope is still of
size $O(n)$ and can be computed in $O(n\log n)$ time
\cite{ref:AgarwalRe90}. Given the query interval $I$,
we can compute in $O(\log n)$ time
the cdf of $\calF$ whose intersection with $L$ is on the
upper envelope of $\calF$.

\begin{theorem}\label{theo:nonuniformtop1}
In the histogram case, we can build in $O(n\log n)$ time an $O(n)$ size
data structure on $P$ that can answer each top-1 query with an unbounded query
interval in $O(\log n)$ time.
\end{theorem}

\subsubsection{Threshold Queries}
For the threshold query, as discussed in Section \ref{sec:pre} we only
need to find the cdfs of $\calF$ whose intersections with $L$
have $y$-coordinates at least $\tau$. Let $q_I$ be the point
$(x_r,\tau)$ on $L$. A line segment is {\em vertically above} $q_I$ if
the segment intersects $L$ and the intersection is at least as high
as $q_I$.
Hence, to answer the threshold query on $I$,
it is sufficient to find the segments of $\calF$ that are
vertically above $q_I$.
Agarwal {\em et al.} \cite{ref:AgarwalIn09} gave the following
result on the {\em segment-below-point queries}.
For a set $S$ of $O(n)$ line segments in the plane, a data structure of
$O(n)$ size can be computed in $O(n\log n)$ time that can report the
segments of $S$ vertically {\em below} a query point $q$ in $O(m'+\log
n)$  time, where $m'$ is the output size.
In our problem, we need a data structure on $\calF$ to solve the {\em
segments-above-point queries}, which can be solved by
using the same approach as \cite{ref:AgarwalIn09}.
Therefore,  we can build in $O(n\log n)$ time an $O(n)$ data structure
on $P$ that can answer each threshold query with an
unbounded query
interval in $O(m+\log n)$ time.

\begin{theorem}\label{theo:nonuniformthreshold}
In the histogram case, we can build in $O(n\log n)$ time an $O(n)$ size
data structure on $P$ that can answer each threshold query with an
unbounded query
interval in $O(m+\log n)$ time, where $m$ is the output size of the query.
\end{theorem}

\subsubsection{Top-k Queries}
For the top-$k$ queries, 
we only need to find the $k$ segments of $\calF$ whose
intersections with $L$ are the highest. To this end, we can slightly
modify the data structure for the segment-below-point queries given in
\cite{ref:AgarwalIn09}.

The data structure in \cite{ref:AgarwalIn09} is a binary tree structure that
maintains a number of sets of lines (each such line contains a
segment of $\calF$). For each such set of lines,
a half-plane range reporting data structure similar to that in Section
\ref{sec:pre} is built, where the
lower envelopes (instead of the upper envelopes as we discussed in Section
\ref{sec:pre}) of the layers of the lines are maintained. For our
purpose, we replace it by our half-plane range reporting data
structure in Section \ref{sec:pre} (i.e., maintain the upper
envelopes). With this modification,
we can answer the segments-above-point
queries in the following way.

Consider a query point $q=(x_r,-\infty)$
(i.e., the lower infinite endpoint of $L$), and suppose we want
to find the segments of $\calF$ vertically above $q$, which are
also the segments intersecting $L$. By using the data structure
\cite{ref:AgarwalIn09} modified as above, the query algorithm works as
follows. First, with the help of fractional cascading,
in $O(\log n)$ time,
the query algorithm will find $O(\log n)$ half-plane range reporting
data structures such that for each such data structure $D$ the segment
intersecting $L$ on the upper envelope of $D$ is
known. Second, for each such half-plane range reporting data structure $D$,
from the above known segment intersecting $L$, by using the fractional
cascading and walking on the upper envelopes of the layers of $D$, we can
report all lines of $D$ higher than $q$ in constant time each. The first step
takes $O(\log n)$ time, and the second step takes $O(m')$ time, where
$m'$ is the total output size.

For our problem, we only need to report the highest $k$ segments of
$\calF$ that are vertically above $q$. To this end,
we will modify the query algorithm such that
the segments of $\calF$ vertically above $q$ will be reported in order from
top to bottom, and once $k$ segments are reported, we will terminate the
algorithm. We use a heap $H$ in a similar way as in Section
\ref{sec:unbounded} for the
uniform case. Specifically, in the first step, we find the $O(\log n)$
half-plane range reporting data structures, and for each such data
structure $H$, the highest segment of $D$
intersecting $L$ is known. In the second step, we
build a heap $H$ on these $O(\log n)$ segments where the keys are the
$y$-coordinates of their intersections with $L$. The segment in
$H$ with the largest key must be the highest segment of $\calF$
intersecting $L$. We remove the segment from $H$, and let $D$
be the half-plane range reporting
data structure that contains the segment. As in Section
\ref{sec:unbounded} for the uniform
case, we determine in constant time at most three segments from $D$ and
insert them to $H$. Now the segment of $H$
with the largest key is the second highest segment of $\calF$
intersecting $L$. We repeat the above procedure until we have
reported $k$ segments.

To analyze the running time, the first step takes $O(\log n)$ time.
In the second step, we have $O(k)$ operations on $H$ and the segments
that are inserted to $H$ can be found in constant time each by the
range-reporting data structures. The size of the heap
$H$ in the entire query algorithm is $O(k+\log n)$.
Hence, the running
time of the query algorithm is $O(k\log(k+\log n)+\log n)$, which is
$O(k\log k+\log n)$ by Lemma \ref{lem:time}.

Similarly, if $k=\Omega(\log n\log\log n)$, we can answer the top-$k$ query in $O(k)$ time, as follows. As discussed above, in $O(\log n)$ time we find the $O(\log n)$
half-plane range reporting data structures, and for each such data
structure $D$, the highest segment of $D$ intersecting $L$ is known. The answer to the top-$k$ query is the highest $k$ intersections of $L$ and the lines in these $O(\log n)$ half-plane range reporting data structures. This is exactly the same situation as in Lemma \ref{lem:100}, where we also have $O(\log n)$ half-plane range reporting data structures. Hence, the algorithm in  Lemma \ref{lem:100} is applicable here, which runs in $O(k)$ time.

In summary, we obtain the following results.

\begin{theorem}\label{theo:nonuniformtopk}
We can build in $O(n\log n)$ time an $O(n)$ size
data structure on $P$ that can answer each top-$k$ query with an unbounded query
interval in $O(k)$ time if $k=\Omega(\log n\log\log n)$ and $O(k\log k+\log n)$ time otherwise.
\end{theorem}

\subsection{Queries with Bounded Intervals}
\label{sec:boundedhist}

In this case, the query interval $I=[x_l,x_r]$ is bounded.

For this case, Agarwal {\em et al.} \cite{ref:AgarwalIn09} built a data structure of size
$O(n\log^2 n)$ in $O(n\log^3 n)$ expected time, which can answer each threshold query in $O(\log^3 n + m)$ time. We first briefly discuss this data structure (refer to Section 4 of \cite{ref:AgarwalIn09} for more details) because our data structures for top-1 and top-$k$ queries also use some of their techniques.

Agarwal {\em et al.} \cite{ref:AgarwalIn09} built a data structure (a binary search tree), denoted by $\calT$, which maintains a family of {\em canonical sets} of planes in 3D (defined by the uncertain points of $P$). Consider any query interval $I=[x_l,x_r]$ with a threshold value $\tau$. Let $q(I)$ be the point with coordinates $(x_l,x_r,\tau)$ in 3D, and let $L(I)$ be the line through $q$ and parallel to the $z$-axis. Using $\calT$, one can determine a family $F(I)$ of $O(\log^2 n)$ canonical sets of $\calT$ with the following property: Each uncertain point $p$ defines one and only one plane in $F(I)$ such that the $z$-coordinate of the intersection of the plane with $L(I)$ is the probability $\Pr[p\in I]$.
Note that the canonical sets of $F(I)$ are pairwise disjoint.

To answer the threshold query on $I$ and $\tau$, it is sufficient to report the planes in each canonical set of $F(I)$ that lie above the point $q(I)$. To this end, for each canonical set $S$ of $\calT$, Agarwal {\em et al.} \cite{ref:AgarwalIn09} constructed a halfspace range-reporting data structure given by Afshani and Chan \cite{ref:AfshaniOp09} on the planes in $S$ in $O(|S|)$ space and $O(|S|\log |S|)$ expected time, such that given any point $q$, one can report the planes of $S$ above $q$ in $O(\log |S|+M)$ time, where $M$ is the output size. In this way, because there are $O(\log^2 n)$ canonical sets in $F(I)$, the threshold query can be answered in $O(\log^3 n+ m)$ time. The total space of $\calT$ including the halfspace range-reporting data structures is $O(n\log^2 n)$ and $\calT$ can be built in $O(n\log^3 n)$ expected time.

\subsubsection{Top-1 Queries}

Consider the top-1 query on the above query interval $I$. To answer the query, it suffices to find the plane in $F(I)$ whose intersection with $L(I)$ is the highest. To this end, it is sufficient to know the intersection of $L(I)$ with the upper envelope of each canonical set of $F(I)$.
Therefore, for each canonical set $S$ of $\calT$, instead of constructing a halfspace range-reporting data structure, we compute the upper envelope of the planes of $S$ \cite{ref:deBergCo08} and build a point location data structure \cite{ref:EdelsbrunnerOp86,ref:KirkpatrickOp83} on the upper envelope, which can be done in $O(|S|)$ space and $O(|S|\log |S|)$ time. In this way, for each canonical set $S$ of $F(I)$, in $O(\log n)$ time we can determine the plane intersecting $L(I)$ in the upper envelope of $S$. Hence, the top-1 query can be answered in $O(\log^3 n)$ time since $F(I)$ has $O(\log^2 n)$ canonical sets.

Comparing with the original data structure in \cite{ref:AgarwalIn09}, since we spend $O(|S|)$ space and $O(|S|\log |S|)$ time on each canonical set $S$ of $\calT$, the entire data structure can be constructed in $O(n\log^2n)$ space and $O(n\log^3 n)$ (deterministic) time. We summarize the result for the top-1 queries in the follow theorem.

\begin{theorem}\label{theo:nonuniformtop1bounded}
In the histogram case, we can build in $O(n\log^3 n)$ time an $O(n\log^2)$ size
data structure on $P$ that can answer each top-1 query with a bounded query
interval in $O(\log^3 n)$ time.
\end{theorem}

\subsubsection{Top-$k$ Queries}

Consider the top-$k$ query on the query interval $I$. To answer the query, it suffices to find the $k$ planes in $F(I)$ whose intersections with $L(I)$ are the highest. To this end, for each canonical set $S$ of $\calT$, we build a {\em $t$-highest plane} data structure given by Afshani and Chan \cite{ref:AfshaniOp09} on the planes of $S$ in $O(|S|)$ space and $O(|S|\log |S|)$ expected time, such that given any integer $t$ and any query line $L$ parallel to the $z$-axis, the $t$ highest planes of $S$ at $L$ can be found in $O(\log |S|+t)$ time. Comparing with the original data structure in \cite{ref:AgarwalIn09}, since we spend asymptotically the same space and time on each canonical set of $\calT$, our data structure can be constructed in $O(n\log^2n)$ space and $O(n\log^3 n)$ expected time.

To answer the top-$k$ query on $I$, one straightforward way works as follows. For each canonical set $S$ of $F(I)$, by using the $t$-highest plane data structure with $t=k$, we compute the highest $k$ planes of $S$ at $L(I)$. Since there are $O(\log^2 n)$ canonical sets in $F(I)$, the above computes $O(k\log^2 n)$ planes, and among them the highest $k$ planes at $L(I)$ are the answer to the top-$k$ query. The query time is $O(\log^2 n(\log n+k))$. In the following, we present an improved query algorithm with time $O(\log^3 n+k)$.

In the following discussion, for simplicity, whenever we refer to the relative order the planes (e.g., highest, lowest, higher, lower), we refer to their intersections with the line $L(I)$. For example, by ``a plane is higher than another plane'', we mean that the first plane has a higher intersection with $L(I)$ than  the second plane. For ease of exposition, we assume the intersection points of the planes of $F(I)$ with $L(I)$ are distinct. Note that $F(I)$ is a family of canonical sets; but by slightly abusing the notation, when we say ``a plane of $F(I)$'', we really mean that the plane is in a canonical set of $F(I)$.

We make use of some idea from Lemma \ref{lem:100} although the details are quite different. Our algorithm has two steps: a main algorithm and a post-processing algorithm. We discuss the main algorithm first.

Let $f=|F(I)|$, and thus $f=O(\log^2 n)$. Let $S_1,S_2,\ldots, S_f$ be the canonical sets of $F(I)$. For each canonical set $S_i$, let $S_i^j$ denote the set of the highest $2^{j-1}\cdot \log n$ planes of $S_i$ for $j=1,2,\ldots$, and we let $S_i^j=\emptyset$ for $j=0$. For each $S_i$,
our main algorithm maintains a subset $S_i'\subseteq S_i$ and an integer $j(i)$, such that $S_i'=S_i^{j(i)}$. We also maintain a max-heap $H$ that contains the lowest plane in the subset $S_i'$ for each $S_i\in F(I)$. Hence, the size of $H$ is $O(\log^2 n)$. The ``keys'' of the planes in $H$ are the $z$-coordinates of their intersections with $L(I)$. In addition, our algorithm maintains an integer $r$, which is the size of a set $R$ of planes.
Before the main algorithm stops, $R=\cup_{i=1}^f S_i^{j(i)-1}$ (after the main algorithm stops, the definition of $R$ is slightly different; see the details below).
Note that our algorithm does not maintain $R$ explicitly, and we use $R$ only to argue the correctness of the algorithm.
During the main algorithm, $r$ will get increased, and the main algorithm stops once $r\geq k$ (at which moment we have identified a set of $O(\log^3 n+k)$ planes, and among them the highest $k$ planes are the answer to our top-$k$ query, which will be found later by the post-processing algorithm).

Initially, for each canonical set $S_i$ of $F(I)$, by using the $t$-highest plane data structure with $t=\log n$, we compute $S_i'=S_i^1$, and further we find the lowest plane in $S_i'$ and insert it into $H$; the above can be done in $O(\log n + |S_i'|)$ time, which is $O(|S_i'|)$ time due to $|S_i'|=|S_i^1|=\log n$. Also, initially we set $r=0$ ($R$ is implicitly set to $\emptyset$),
and set $j(i)=1$ for each $i$ with $1\leq i\leq f$.

Next, we do an ``extract-max'' operation on $H$ to find the highest plane in $H$ and remove it from $H$. Suppose the above plane is from a canonical set $S_i$ for some $i$. Then, we let $R=R\cup S_i'$ and set $r=r+|S_i'|$. Further, by using the $t$-highest plane data structure with $t=2\log n $, we compute $S_i'=S_i^2$, and then we find the lowest plane in $S_i'$ and insert it into $H$; again, the above can be done in $O(|S_i'|)$ time. Finally, we update $j(i)=2$.

In general, we do an extract-max operation on the current $H$ and suppose the removed plane is from a canonical set $S_i$ for some $i$. 
We let $R=R\cup S_i^{j(i)}\setminus S_i^{j(i)-1}$ (note that $S_i^{j(i)-1}\subseteq S_i^{j(i)}$, and thus this just adds those planes of $S_i^{j(i)}$ that are not in $S_i^{j(i)-1}$ to $R$), and set $r=r+|S_i^j|-|S_i^{j-1}|$. Again, we do not explicitly maintain $R$ but explicitly maintain $r$. If $r\geq k$, then we stop the main algorithm. Otherwise, by using the $t$-highest plane data structure with $t=2^{j(i)}\cdot \log n $, we compute $S_i'=S_i^{j(i)+1}$ (and the previous $S_i'$ is discarded), and further we find the lowest plane in $S_i^{j(i)+1}$ and insert it into $H$; again, the above can be done in $O(|S_i'|)$ time. Note that for ease of exposition, we assume $|S_i|\geq 2^{j(i)}\cdot \log n$ (otherwise, we can solve the problem by similar techniques with more tedious discussions). Finally, we increase $j(i)$ by one.

The above finishes the main algorithm. After it stops, let $R'=\cup_{i=1}^f S_i^{j(i)}$. Let $B$ be the set of $k$ highest planes in $F(I)$, i.e., $B$ is the answer to our top-$k$ query on $I$.
We have the following lemma.

\begin{lemma}\label{lem:120}
$R\subseteq R'$, $B\subseteq R'$, and $|R'|= O(\log^3 n+k)$.
\end{lemma}
\begin{proof}
We first show $R\subseteq R'$.

Suppose the last extract-max operation on $H$ in the main algorithm removes a plane from the canonical set $S_{i}$ for $i=a$. Hence, the algorithm stops after we have $R=R\cup S_a^{j(a)}\setminus S_a^{j(a)-1}$. Thus, $S_a^{j(a)}\subseteq R$. Consider any other canonical set $S_i$ with $i\neq a$. According to our algorithm, it always holds that $S_i^{j(i)-1}\subseteq R$.
Therefore, we have the following:
$$R=S_a^{j(a)}\cup\bigcup_{1\leq i\leq f, i\neq a}S_i^{j(i)-1}\text{\ and \ } R'=\bigcup_{1\leq i\leq f}S_i^{j(i)}.$$

Since $S_i^{j(i)-1}\subseteq S_i^{j(i)}$ for any $i$, we obtain that $R\subseteq R'$.

Next, we show that $|R'|=O(\log^3 n+k)$.

Indeed, since the algorithm stops right after $r=r+|S_a^{j(a)}|-|S_a^{j(a)-1}|\geq k$ and $R=R\cup S_a^{j(a)}\setminus S_a^{j(a)-1}$, the original value of $r$ before the above increasing is less than $k$. In other words, $\sum_{1\leq i\leq f}|S_i^{j(i)-1}|<k$.
For each $S_i$ of $F(I)$, if $j(i)=1$, then $|S_i^{j(i)}|=\log n$ and $|S_i^{j(i)-1}|=0$; otherwise, $|S_i^{j(i)}|=2|S_i^{j(i)-1}|$. Therefore, we obtain $|R'|=\cup_{1\leq i\leq f}S_i^{j(i)}\leq 2\cdot \sum_{1\leq i\leq f}|S_i^{j(i)-1}| + f\cdot \log n = 2k+O(\log^3 n)=O(\log^3 n +k)$ because $f=|F(I)|=O(\log^2 n)$.

Finally, we prove $B\subseteq R'$.

Let $\sigma^*$ denote the plane removed by the last extract-max operation on $H$ in the main algorithm.
We claim that $\sigma^*$ is the lowest plane in $R$. We prove the claim below.

According to our algorithm, the planes removed by the extract-max operations on $H$ follow the order from high to low. Consider any plane $\sigma\in R$. To prove the claim, it is sufficient to show that $\sigma^*$ is not higher than $\sigma$. According to our algorithm, the first time $\sigma$ is added in $R$ must be due to an operation on $R$: $R=R\cup S_i^{j(i)}\setminus S_i^{j(i)-1}$ after an extract-max operation removes a plane $\sigma'$ from $H$ and $\sigma'$ is from a canonical set $S_i$. This implies that $\sigma\in S_i^{j(i)}\setminus S_i^{j(i)-1}$. According to our algorithm, $\sigma'$ is the lowest plane in the above $S_i^{j(i)}$, and thus, $\sigma'$ is not higher than $\sigma$. On the other hand, since $\sigma^*$ is the last plane removed by the extract-min operations, $\sigma^*$ is not higher than $\sigma'$. Therefore, $\sigma^*$ is not higher than $\sigma$, and the above claim is proved.

Consider any plane $\sigma\in B$. To show $B\subseteq R'$, it suffices to prove $\sigma\in R'$.
If $\sigma$ is in $R$, then since $R\subseteq R'$, $\sigma\in R'$ is true. Below we assume $\sigma\not\in R$, and thus $\sigma\neq \sigma^*$.

Note that $\sigma\in B$ implies that there are at most $k-1$ planes of $F(I)$ higher than $\sigma$. Since $r=|R|\geq k$ and $\sigma^*$ is the lowest plane in $R$, $\sigma$ must be higher than $\sigma^*$ since otherwise all planes in $R$ would be higher than $\sigma$, contradicting with that there are at most $k-1$ planes higher than $\sigma$.

Assume $\sigma$ is in a canonical set $S_i$ for some $i$.
Recall that $\sigma^*$ is from the canonical set $S_a$.
Note that all planes of $S_a$ higher than $\sigma^*$ are in $S_a^{j(a)}$.
By our definition of $R$, $S_a^j(a)\subseteq R$.
Since $\sigma$ is higher than $\sigma^*$ and $\sigma\not\in R$, we can obtain $i\neq a$.
According to our algorithm, after the algorithm stops, $H$ contains a plane $\sigma_i$ from $S_i$, and $\sigma_i$ is the lowest plane in $S_i^{j(i)}$.
Recall that, $S_i^{j(i)}\subseteq R'$. This implies that all planes of $S_i$ higher than $\sigma_i$ are in $R'$. Since $\sigma^*$ is removed by an extract-min operation and after the operation $\sigma_i$ is still in $H$, $\sigma^*$ must be higher than $\sigma_i$. Because $\sigma$ is higher than $\sigma^*$, $\sigma$ is higher than $\sigma_i$.

In summary, the above discussion obtains the following: $\sigma$ is in $S_i$; $\sigma$ is higher than $\sigma_i$; all planes of $S_i$ higher than $\sigma_i$ are in $R'$. Thus, we obtain that $\sigma$ is in $R'$.

Therefore, we conclude that $B\subseteq R'$, and the lemma follows.
\qed
\end{proof}

Based on Lemma \ref{lem:120}, if we have the set $R'$ explicitly, then we can compute $B$ in additional $O(\log^3 n+k)$ time by using the linear time selection algorithm \cite{ref:CLRS09}. However, the above main algorithm does not explicitly compute $R'$, but it has maintained $j(i)$ for each $S_i\in F(I)$. Since $R'=\cup_{1\leq i\leq f}S_i^{j(i)}$, we can compute $R'$ by using a $t$-highest plane query with $t=2^{j(i)-1}\cdot \log n$ on each canonical set $S_i$ of $F(I)$.

The following lemma gives the running time of our entire top-$k$ query algorithm.

\begin{lemma}
The time complexity of our top-$k$ query algorithm is $O(\log^3 n+k)$.
\end{lemma}
\begin{proof}
We first analyze the main algorithm, whose running time mainly depends on the time of the $t$-highest plane queries and the time of the operations on the heap $H$. We first give a bound on the time of the $t$-highest plane queries, with the help of Lemma \ref{lem:120}.

Note that after each $t$-highest plane query in the main algorithm, we always find the lowest plane in the output planes of the query, whose time is only linear to the number of output planes and is upper bounded by the above query time. In the following, we focus on analyzing the time of the $t$-highest plane queries.

Consider any canonical set $S_i\in F(I)$. According to our algorithm, for each $j$ with $1\leq j\leq j(i)$, the main algorithm performs a $t$-highest plane query on $S_i$ with $t=2^{j-1}\cdot \log n$ to compute $S_i^j$, which takes $O(|S_i^j|)$ time (we ignore the $\log n$ factor in the query time because $\log n\leq |S_i^j|$). Hence, the total time of the $t$-highest plane queries on $S_i$ in the main algorithm is $O(\sum_{j=1}^{j(i)}|S_i^j|)$. Note that
$$\sum_{j=1}^{j(i)}|S_i^j|=\log n\cdot \sum_{j=1}^{j(i)}2^{j-1}\leq \log n\cdot 2^{j(i)}=2\cdot |S_i^{j(i)}|.$$

Recall that  $|R'|=\sum_{1\leq i\leq f}|S_i^{j(i)}|$. Hence, the total time on the $t$-highest plane queries in the entire main algorithm is $O(|R'|)$, which is $O(\log^3 n + k)$ by Lemma \ref{lem:120}.

Next, we analyze the time we spent on the heap $H$. Recall that the size of $H$ is $O(\log^2 n)$. Initially, we build $H$ on $O(\log^2 n)$ planes, which can be done in $O(\log^2 n)$ time.
Later in the algorithm, the operations on $H$ include the extract-max and insertion operations.
We need to figure out how many operations were performed on $H$ in the main algorithm.

Consider any extract-max operation on $H$, and suppose the removed plane is from set $S_i$. Then, after the operation, we have $R=R\cup S_i^{j(i)}\setminus S_i^{j(i)-1}$, and  since $|S_i^{j(i)}|-|S_i^{j(i)-1}|\geq \log n$, the above increases $R$ by at least $\log n$ planes. After that, there is at most one insertion operation on $H$. Since the main step stops once $|R|\geq k$, the total number of extract-max operations is at most $\frac{k}{\log n}$. The number of insertion operations is also at most $\frac{k}{\log n}$. Since $|H|=O(\log^2 n)$, each operation on $H$ takes $O(\log\log n)$ time. The total time on $H$ is $O(\log^2 n+\frac{k}{\log n}\cdot \log\log n)=O(\log^2 n+ k)$.

Therefore, the total time of the main algorithm is $O(\log^3 n + k)$.

Finally, we analyze the running time of the post-processing step, which computes $R'$ and finds the highest $k$ planes in $R'$. Computing $R'$ is done by doing a $t$-highest plane query with $t=j(i)$ on each set $S_i$. Therefore, as above, the total time is at most $|R'|$, which is $O(\log^3 n + k)$. Finding the highest $k$ planes in $R'$ takes $O(|R'|)$ time by using the linear time selection algorithm \cite{ref:CLRS09}.

Thus, the total time of our top-$k$ query algorithm is $O(\log^3 n + k)$. \qed
\end{proof}

The above discussion leads to the following theorem.

\begin{theorem}\label{theo:nonuniformtopkbounded}
We can build in $O(n\log^3 n)$ expected time an $O(n\log^2 n)$ size
data structure on $P$ that can answer each top-$k$ query with a bounded query
interval in $O(\log^3 n+ k)$ time.
\end{theorem}

Note that the planes reported by our top-$k$ query algorithm are not in any sorted order.

\section{Conclusions}
\label{sec:conclusions}

In this paper we present a number of data structures for answering a variety of range queries over
uncertain data in one dimensional space. In general, our data structures have linear or nearly linear sizes and can support efficient queries. While it would be interesting to develop better solutions,
an interesting but challenging open problem is whether we can generalize
our techniques to solve the corresponding problems
in higher dimensions, for which only heuristic results have been proposed
\cite{ref:TaoIn05,ref:TaoRa07}.


\section*{Acknowledgments}


The authors would like to thank Sariel Har-Peled for several insightful comments, and for suggesting us the canonical set idea (somewhat similar to that in \cite{ref:AgarwalIn09}), which leads us to the solutions for the histogram bounded case of the problem.



\bibliographystyle{plain}

\begin{thebibliography}{10}

\bibitem{ref:AbdullahRa13}
A.~Abdullah, S.~Daruki, and J.M. Phillips.
\newblock Range counting coresets for uncertain data.
\newblock In {\em the 29th Symposium on Computational Geometry}, pages
  223--232, 2013.

\bibitem{ref:AfshaniOp09}
P.~Afshani and T.M. Chan.
\newblock Optimal halfspace range reporting in three dimensions.
\newblock In {\em Proc. of the 20th Annual ACM-SIAM Symposium on Discrete
  Algorithms (SODA)}, pages 180--186, 2009.

\bibitem{ref:AgarwalIn09}
P.K. Agarwal, S.-W. Cheng, Y.~Tao, and K.~Yi.
\newblock Indexing uncertain data.
\newblock In {\em Proc. of the 28th Symposium on Principles of Database
  Systems}, pages 137--146, 2009.

\bibitem{ref:AgarwalNe12}
P.K. Agarwal, A.~Efrat, S.~Sankararaman, and W.~Zhang.
\newblock Nearest-neighbor searching under uncertainty.
\newblock In {\em Proc. of the 31st Symposium on Principles of Database Systems
  (PODS)}, pages 225--236, 2012.

\bibitem{ref:AgarwalRe90}
P.K. Agarwal and M.~Sharir.
\newblock Red-blue intersection detection algorithms, with applications to
  motion planning and collision detection.
\newblock {\em SIAM Journal on Computing}, 19:297--321, 1990.

\bibitem{ref:AgrawalTr06}
P.~Agrawal, O.~Benjelloun, A.~Das Sarma, C.~Hayworth, S.~Nabar, T.~Sugihara,
  and J.~Widom.
\newblock Trio: a system for data, uncertainty, and lineage.
\newblock In {\em Proc. of the 32nd International Conference on Very Large Data
  Bases}, pages 1151--1154, 2006.

\bibitem{ref:BeskalesEf08}
G.~Beskales, M.A. Soliman, and I.F. IIyas.
\newblock Efficient search for the top-k probable nearest neighbors in
  uncertain databases.
\newblock In {\em Proc. of the VLDB Endowment}, pages 326--339, 2008.

\bibitem{ref:ChazelleFi86}
B.~Chazelle.
\newblock Filtering search: A new approach to query-answering.
\newblock {\em SIAM Journal on Computing}, 15(3):703--724, 1986.

\bibitem{ref:ChazelleFr86}
B.~Chazelle and L.J. Guibas.
\newblock Fractional cascading: {I. A} data structuring technique.
\newblock {\em Algorithmica}, 1(1):133--162, 1986.

\bibitem{ref:ChazelleFr862}
B.~Chazelle and L.J. Guibas.
\newblock Fractional cascading: {II. Applications}.
\newblock {\em Algorithmica}, 1(1):163--191, 1986.

\bibitem{ref:ChazelleTh85}
B.~Chazelle, L.J. Guibas, and D.T. Lee.
\newblock The power of geometric duality.
\newblock {\em BIT}, 25:76--90, 1985.

\bibitem{icde08-probnn}
R.~Cheng, J.~Chen, M.~Mokbel, and C.~Chow.
\newblock Probabilistic verifiers: Evaluating constrained nearest-neighbor
  queries over uncertain data.
\newblock In {\em IEEE International Conference on Data Engineering (ICDE)},
  pages 973--982, 2008.

\bibitem{cheng:sigmod03}
R.~Cheng, D.~Kalashnikov, and S.~Prabhakar.
\newblock Evaluating probabilistic queries over imprecise data.
\newblock In {\em ACM SIGMOD International Conference on Management of Data
  (SIGMOD)}, pages 551--562, 2003.

\bibitem{ref:ChengEf04}
R.~Cheng, Y.~Xia, S.~Prabhakar, R.~Shah, and J.S. Vitter.
\newblock Efficient indexing methods for probabilistic threshold queries over
  uncertain data.
\newblock In {\em Proc. of the 30th International Conference on Very Large Data
  Bases}, pages 876--887, 2004.

\bibitem{ref:CLRS09}
T.~Cormen, C.~Leiserson, R.~Rivest, and C.~Stein.
\newblock {\em Introduction to Algorithms}.
\newblock MIT Press, 3nd edition, 2009.

\bibitem{Cormode09}
G.~Cormode, F.~Li, and K.~Yi.
\newblock Semantics of ranking queries for probabilistic data and expected
  ranks.
\newblock In {\em IEEE International Conference on Data Engineering (ICDE)},
  pages 305--316, 2009.

\bibitem{cormode2008approximation}
G.~Cormode and A.~McGregor.
\newblock Approximation algorithms for clustering uncertain data.
\newblock In {\em ACM SIGMOD-SIGACT-SIGART Symposium on Principles of Database
  Systems (PODS)}, pages 191--200. ACM, 2008.

\bibitem{dalvi:vldbj06}
N.~Dalvi and D.~Suciu.
\newblock Efficient query evaluation on probabilistic databases.
\newblock {\em The VLDB Journal}, pages 523--544, 2006.

\bibitem{dalvi2012dichotomy}
N.~Dalvi and D.~Suciu.
\newblock The dichotomy of probabilistic inference for unions of conjunctive
  queries.
\newblock {\em Journal of the ACM}, 59(6):30, 2012.

\bibitem{ref:deBergCo08}
M.~de~Berg, O.~Cheong, M.~van Kreveld, and M.~Overmars.
\newblock {\em Computational Geometry --- Algorithms and Applications}.
\newblock Springer-Verlag, Berlin, 3rd edition, 2008.

\bibitem{ref:DriscollMa89}
J.~Driscoll, N.~Sarnak, D.~Sleator, and R.E. Tarjan.
\newblock Making data structures persistent.
\newblock {\em Journal of Computer and System Sciences}, 38(1):86--124, 1989.

\bibitem{ref:EdelsbrunnerOp86}
H.~Edelsbrunner, L.~Guibas, and J.~Stolfi.
\newblock Optimal point location in a monotone subdivision.
\newblock {\em SIAM Journal on Computing}, 15(2):317--340, 1986.

\bibitem{ref:FredericksonTh82}
G.~Frederickson and D.~Johnson.
\newblock The complexity of selection and ranking in {$X+Y$} and matrices with
  sorted columns.
\newblock {\em Journal of Computer and System Sciences}, 24(2):197--208, 1982.

\bibitem{guha2009exceeding}
S.~Guha and K.~Munagala.
\newblock Exceeding expectations and clustering uncertain data.
\newblock In {\em ACM SIGMOD-SIGACT-SIGART Symposium on Principles of Database
  Systems (PODS)}, pages 269--278. ACM, 2009.

\bibitem{jampani2008mcdb}
R.~Jampani, F.~Xu, M.~Wu, L.L. Perez, C.~Jermaine, and P.J. Haas.
\newblock {MCDB: a monte carlo approach to managing uncertain data}.
\newblock In {\em ACM SIGMOD International Conference on Management of Data
  (SIGMOD)}, pages 687--700. ACM, 2008.

\bibitem{DBLP:conf/pods/JayramMMV07}
T.~S. Jayram, A.~McGregor, S.~Muthukrishnan, and E.~Vee.
\newblock Estimating statistical aggregates on probabilistic data streams.
\newblock In {\em ACM SIGMOD-SIGACT-SIGART Symposium on Principles of Database
  Systems (PODS)}, pages 243--252, 2007.

\bibitem{ref:KirkpatrickOp83}
D.~Kirkpatrick.
\newblock Optimal search in planar subdivisions.
\newblock {\em SIAM Journal on Computing}, 12(1):28--35, 1983.

\bibitem{ref:KnightEf11}
A.~Knight, Q.~Yu, and M.~Rege.
\newblock Efficient range query processing on uncertain data.
\newblock In {\em Proce of IEEE International Conference on Information Reuse
  and Integration}, pages 263--268, 2011.

\bibitem{koch2009maybms}
C.~Koch.
\newblock {MAYBMS:} a system for managing large probabilistic databases.
\newblock {\em Managing and Mining Uncertain Data}, pages 149--183, 2009.

\bibitem{li2010ranking}
J.~Li and A.~Deshpande.
\newblock Ranking continuous probabilistic datasets.
\newblock {\em Proceedings of the VLDB Endowment}, 3(1):638--649.

\bibitem{li2011unified}
J.~Li, B.~Saha, and A.~Deshpande.
\newblock A unified approach to ranking in probabilistic databases.
\newblock {\em The VLDB Journal}, 20(2):249--275, 2011.

\bibitem{ref:LiRa14}
J.~Li and H.~Wang.
\newblock Range queries on uncertain data.
\newblock In {\em Proc. of the 25th International Symposium on Algorithms and
  Computation (ISAAC)}, pages 326--337, 2014.

\bibitem{ref:LjosaAP07}
V.~Ljosa and A.~K. Singh.
\newblock {APLA: Indexing} arbitrary probability distributions.
\newblock In {\em IEEE International Conference on Data Engineering (ICDE)},
  pages 946--955, 2007.

\bibitem{ref:MitchellL192}
J.S.B. Mitchell.
\newblock {$L_1$} shortest paths among polygonal obstacles in the plane.
\newblock {\em Algorithmica}, 8(1):55--88, 1992.

\bibitem{potamias2010k}
M.~Potamias, F.~Bonchi, A.~Gionis, and G.~Kollios.
\newblock {$k$-Nearest} neighbors in uncertain graphs.
\newblock {\em Proceedings of the VLDB Endowment}, 3(1):997--1008, 2010.

\bibitem{ref:QiTh10}
Y.~Qi, R.~Jain, S.~Singh, and S.~Prabhakar.
\newblock Threshold query optimization for uncertain data.
\newblock In {\em ACM SIGMOD International Conference on Management of Data
  (SIGMOD)}, pages 315--326, 2010.

\bibitem{re:icde07}
C.~Re, N.~Dalvi, and D.~Suciu.
\newblock Efficient top-k query evaluation on probabilistic data.
\newblock In {\em IEEE International Conference on Data Engineering (ICDE)},
  pages 886--895, 2007.

\bibitem{conf/dbpl/re07}
C.~R\'{e} and D.~Suciu.
\newblock Efficient evaluation of {HAVING} queries on a probabilistic database.
\newblock In {\em Proc. of the 11th International Conference on Database
  Programming Languages}, pages 186--200, 2007.

\bibitem{sen:vldbj09}
P.~Sen, A.~Deshpande, and L.~Getoor.
\newblock Prdb: managing and exploiting rich correlations in probabilistic
  databases.
\newblock {\em The VLDB Journal}, 18(5):1065--1090, 2009.

\bibitem{ref:SinghIn07}
S.~Singh, C.~Mayfield, S.~Prabhakar, R.~Shah, and S.E. Hambrusch.
\newblock Indexing uncertain categorical data.
\newblock In {\em IEEE International Conference on Data Engineering (ICDE)},
  pages 616--625, 2007.

\bibitem{soliman:icde07}
M.A. Soliman, I.F. Ilyas, and C.C. Chang.
\newblock {Top-k query processing in uncertain databases}.
\newblock In {\em IEEE International Conference on Data Engineering (ICDE)},
  pages 896--905. IEEE, 2007.

\bibitem{suciu2011probabilistic}
D.~Suciu, D.~Olteanu, C.~R{\'e}, and C.~Koch.
\newblock Probabilistic databases.
\newblock {\em Synthesis Lectures on Data Management}, 3(2):1--180, 2011.

\bibitem{ref:TaoIn05}
Y.~Tao, R.~Cheng, X.~Xiao, W.K. Ngai, B.~Kao, and S.~Prabhakar.
\newblock Indexing multi-dimensional uncertain data with arbitrary probability
  density functions.
\newblock In {\em Proc. of the 31st International Conference on Very Large Data
  Bases (VLDB)}, pages 922--933, 2005.

\bibitem{ref:TaoRa07}
Y.~Tao, X.~Xiao, and R.~Cheng.
\newblock Range search on multidimensional uncertain data.
\newblock {\em ACM Transactions on Database Systems (TODS)}, 32, 2007.

\bibitem{ref:YiuEf09}
M.L. Yiu, N.~Mamoulis, X.~Dai, Y.~Tao, and M.~Vaitis.
\newblock Efficient evaluation of probabilistic advanced spatial queries on
  existentially uncertain data.
\newblock {\em IEEE Transactions of Knowledge Data Engineering (TKDE)},
  21:108--122, 2009.

\end{thebibliography}

%



\end{document}